\documentclass{llncs}
\usepackage{graphicx}
\usepackage{latexsym}
\usepackage{amsmath}
\usepackage{amsfonts}
\usepackage[psamsfonts]{amssymb}
\newcommand{\hide}[1]{}

\newcommand{\dout}{{d^+}}
\newcommand{\doutG}{{d_G^+}}
\newcommand{\din}{{d^-}}
\newcommand{\dinG}{{d_G^-}}
\newcommand{\Nout}{{N^+}}
\newcommand{\NoutG}{{N_G^+}}
\newcommand{\Nin}{{N^-}}
\newcommand{\NinG}{{N_G^-}}
\newcommand{\pw}{{\rm pw}}
\newcommand{\wild}{{\rm wld}}

\newcommand{\prob}{{\rm\bf Pr}}
\newcommand{\ex}{{\rm\bf E}}

\newcommand{\Vinle}[1]{V^-_{\le {#1}}}
\newcommand{\Vinge}[1]{V^-_{\ge {#1}}}
\newcommand{\Voutle}[1]{V^+_{\le {#1}}}
\newcommand{\Voutge}[1]{V^+_{\ge {#1}}}
\begin{document}
\title{On the Pathwidth of \\Almost Semicomplete Digraphs}

\author{Kenta Kitsunai\inst{1}\and
Yasuaki Kobayashi\inst{2} \and Hisao Tamaki\inst{3}}
\authorrunning{K.~Kitsunai, Y~Kobayashi, and H.~Tamaki}
\institute{NTT DATA Corporation\\\email{mizuna0719@gmail.com}\and
Computer Center, Gakushuin University\\
\email{yasuaki.kobayashi@gakushuin.ac.jp} \and
Department of Computer Science, Meiji University\\
\email{tamaki@cs.meiji.ac.jp}}

\maketitle

\begin{abstract}
We call a digraph {\em $h$-semicomplete} if
each vertex of the digraph has at most $h$ non-neighbors, where
a non-neighbor of a vertex $v$ is a vertex $u \neq v$ such that 
there is no edge between $u$ and $v$ in either direction.
This notion generalizes
that of semicomplete digraphs which are $0$-semicomplete
and tournaments which are semicomplete and have no anti-parallel pairs
of edges.
Our results in this paper are as follows.
(1) We give an algorithm which, given an $h$-semicomplete
digraph $G$ on $n$ vertices and a positive integer $k$,
in $(h + 2k + 1)^{2k} n^{O(1)}$ time either constructs 
a path-decomposition of $G$ of width at most $k$ or concludes
correctly that the pathwidth of $G$ is larger than $k$.
(2) We show that there is a
function $f(k, h)$ such that every $h$-semicomplete digraph of 
pathwidth at least $f(k, h)$
has a semicomplete subgraph of pathwidth at least $k$.

One consequence of these results is that the problem of
deciding if a fixed digraph $H$ is topologically contained in a given
$h$-semicomplete digraph $G$ admits a polynomial-time algorithm
for fixed $h$.
\end{abstract}

\section{Introduction}
A {\em tournament} is a digraph obtained from a complete graph by
orienting each edge. A {\em semicomplete digraph} generalizes
a tournament, allowing each pair of distinct vertices to optionally
have two edges in both directions between them.
Tournaments and semicomplete digraphs are well-studied (see \cite{B-JG08},
for example) and have recently been attracting renewed interests in the
following context.

There are many problems on undirected graphs that admit polynomial
time algorithms but have digraph counterparts that are NP-complete.
For example, Robertson and Seymour \cite{RS95}, in their Graph Minors project,
proved that the $k$ disjoint paths problem (and the $k$ edge-disjoint paths
problem) can be solved in polynomial for fixed $k$. On the other hand, 
digraph versions of these problems are NP-complete even for $k = 2$ due to  
Fortune, Hopcroft, and Wyllie \cite{FHW80}.
Recently,
Chudnovsky, Scot, and Seymour \cite{CSS15} showed that the $k$ directed
disjoint paths problem can be solved in polynomial time for fixed $k$ if the
digraph is restricted to be semicomplete. The edge-disjoint version of the problem is
also polynomial time solvable on semicomplete digraphs, due to 
Fradkin and Seymour \cite{FS15}.
The situation is similar for the topological containment problem, which asks if a given graph (digraph)
contains a subgraph isomorphic to a subdivision of a fixed graph (digraph) $H$:
the undirected version is polynomial time solvable due to the disjoint paths
result and the directed version is NP-complete on general digraphs \cite{FHW80}, 
while the question on semicomplete digraphs is polynomial time solvable due
to Fradkin and Seymour \cite{FS13} and moreover is fixed-parameter
tractable due to Fomin and Pilipczuk \cite{FP13,Pilipczuk12}.
In addition to these algorithmic results, 
some well-quasi-order results that are similar to 
the celebrated 
Graph Minors theorem of Robertson and Seymour \cite{RS04} have
been proved on the class of semicomplete digraphs \cite{CS11,KS15}.
These developments seem to suggest that 
the class of semicomplete digraphs is a promising stage for pursuing
digraph analogues of the splendid outcomes,
direct and indirect, from the Graph Minors project.

Given this progress on semicomplete digraphs, it is natural to look for
more general classes of digraphs on which similar results hold.
Indeed, the results on disjoint paths problems cited above are proved for
some generalizations of semicomplete digraphs.  The
vertex-disjoint path algorithm given in \cite{CSS15} works for a digraph class 
called $d$-path dominant digraphs, which
contains semicomplete digraphs ($d = 1$) and 
digraphs with multipartite underlying graphs ($d = 2$).
The edge-disjoint path algorithm given in \cite{FS15} works for
digraphs with independence number (of the underlying graph) bounded by some
fixed integer. On the other hand, the results for
topological containment in \cite{FS13,FP13,Pilipczuk12} are strictly for the
class of semicomplete graphs.

The {\em pathwidth} of digraphs, which plays an essential role in
some of the above results, is defined as follows. 
Let $G$ be a digraph. A {\em path-decomposition} of $G$ is a sequence $(X_1,
\ldots, X_m)$ of vertex sets $X_i \subseteq V(G)$, called {\em bags}, such that
the following three conditions are satisfied:
\begin{enumerate}
  \item $\bigcup_{1 \leq i \leq m} X_i = V(G)$,
  \item for each edge $(u, v)$ of $G$, $u \in X_i$ and $v \in X_j$ for
  some $i \geq j$, and
  \item for every $v \in V(G)$, the set $\{i \mid v \in X_i\}$ of indices 
  of the bags containing $v$ forms a single integer interval.
\end{enumerate}
The first and the third conditions are the same as in the definition of
the pathwidth of undirected graphs; the second condition, on each edge,
is different and depends on the direction of the edge.
Note that some authors, including the present authors in previous work
in different contexts, reverse the direction of edges in this condition. 
We follow the convention of the papers cited above.
As in the case of undirected graphs, the {\em width} of a
path-decomposition $(X_1, \ldots, X_m)$ is 
$\max_{1 \leq i \leq m}|X_i| - 1$ and the {\em pathwidth} of $G$, denoted by
$\pw(G)$, is the smallest integer $k$ such that there is a path-decomposition of $G$ of width $k$.

Unlike for the pathwidth of undirected graphs, which is
linear-time fixed-parameter tractable \cite{Bod96}, no FPT-time algorithm is
known for computing the pathwidth of general digraphs: 
only XP-time algorithms (of running time $n^{O(k)}$) are
known. The third author of the current paper proposed one in \cite{Tam11},
which was unfortunately flawed and has recently been 
corrected in \cite{KKKTT15} by the current and two more authors. 
Another XP algorithm is
due to Nagamochi~\cite{Naga12}, which is formulated for a more general problem
of optimizing linear layouts in submodular systems.  

In this paper, we consider another direction of generalizing semicomplete
digraphs and study the pathwidth of digraphs in the generalized class.  
For non-negative integer $h$, we say that a simple digraph $G$ is $h$-semicomplete 
if each vertex of $G$ has at most $h$ non-neighbors,
where a non-neighbor of vertex $v$ is a vertex $u$ distinct from
$v$ such that there is no edge of $G$ between $u$ and $v$ in either direction.
Thus, semicomplete digraphs are 0-semicomplete.  Our main results are as
follows.

\begin{theorem}
\label{thm:alg}
There is an algorithm which, 
given an $h$-semicomplete
digraph $G$ on $n$ vertices and a positive integer $k$,
in $(h + 2k + 1)^{2k} n^{O(1)}$ time either constructs 
a path-decomposition of $G$ of width at most $k$ or concludes
correctly that the pathwidth is larger than $k$.
\end{theorem}
This theorem generalizes the $k^{O(k)} n^2$ time result of Pilipczuk
\cite{Pilipczuk12} on semicomplete digraphs. Compared on semicomplete
digraphs, his algorithm has smaller dependence on $n$ (our 
$O(1)$ exponent on $n$ is naively 4), while  
the hidden constant in the exponent on $k$ can be large.
\begin{theorem}
\label{thm:comb}
There is a function $f(h, k)$ on positive integers $h$ and $k$ such that
each $h$-semicomplete digraph with pathwidth at least
$f(h, k)$ has a semicomplete subgraph of pathwidth at least $k$.
\end{theorem}

The topological containment result in \cite{FS13} is based
on two components.  One is a combinatorial result
that, for each fixed digraph $H$, there is a positive
integer $k$ such that every semicomplete digraph $G$ of
pathwidth larger than $k$ topologically contains $H$.
The second component is a dynamic programming
algorithm that, given a digraph $G$ on $n$ vertices together with a
path-decomposition of width $k$ and a digraph $H$ on $r$ 
vertices with $s$ edges,
decides if $G$ topologically contains $H$ in
$O(n^{3(k + rs) + 4})$ time.  Note that
this algorithm does not require $G$ to be semicomplete.
Theorem~\ref{thm:comb} enables us to generalize the
first component to $h$-semicomplete digraphs and
Theorem~\ref{thm:alg} gives us the path-decomposition
to be used in the dynamic programming.  Thus, we have
the following theorem. 

\begin{theorem}
For fixed positive integer $h$ and fixed digraph $H$, 
the problem of deciding if a given $h$-semicomplete digraph topologically
contains $H$ can be solved in polynomial time. 
\end{theorem}
We should remark that extending the FPT result of \cite{FP13,Pilipczuk12}
in this direction using the approach of this paper appears difficult, 
as the FPT-time dynamic programming algorithm therein heavily relies on the 
strict semicompleteness of the input digraph.

\subsubsection{Techniques}
Our algorithm in Theorem~\ref{thm:alg} borrows the notion of
separation chains from \cite{Pilipczuk12} but the algorithm itself
is completely different from the one in \cite{Pilipczuk12}.
The advantage of our algorithm is that it works
correctly on general digraphs, in contrast to the one in \cite{Pilipczuk12}
which is highly specialized for semicomplete digraphs.
We need a property of $h$-semicomplete digraphs only
in the analysis of the running time.

Our algorithm is based on the one due to Nagamochi~\cite{Naga12} for
more general problem of finding an optimal linear layout for submodular
systems. Informally, his algorithm applied to the pathwidth computation works
as follows. Fix digraph $G$ and let $d^+(U)$ for each $U \subseteq V(G)$ denote
the number of out-neighbors of $U$. The {\em width} of permutation $\pi$ of
$V(G)$ is defined to be the maximum of $d^+(V(\pi'))$ where
$\pi'$ ranges over all the prefixes of $\pi$ and $V(\pi')$ denotes the
set of vertices in $\pi'$. The smallest integer $k$ such that
there is a permutation of width $k$ is called the {\em vertex separation number}
of $G$ and is equal to the pathwidth of $G$ \cite{YC08}.
Thus, our goal is to decide, given $k$, if there is a permutation of $V(G)$ of
width at most $k$.

Nagamochi's algorithm is a combination of divide-and-conquer and
branching from both sides of the permutation. For disjoint subsets $S$ and $T$ of $V(G)$, call
a permutation $\pi$ of $V(G$) an {\em $(S, T)$-permutation}, if it has a
prefix $\pi'$ with $V(\pi') = S$ and a suffix $\pi''$ with $V(\pi'') = T$.
A vertex set $X$ that minimize $d^+(X)$ subject to $S \subseteq X \subseteq
V(G) \setminus T$ is called a minimum {\em $(S, T)$-separator}. A crucial
observation, based on the submodularity of set function $d^+$ is the following.
Let $X$ be a minimum $(S, T)$-separator. Then, if there is an $(S, T)$-permutation of width 
at most $k$ then 
there is such a permutation that is an $(S, V(G) \setminus X)$-permutation and
an $(X, T)$-permutation at the same time.
Thus if there is a minimum $(S, T)$-separator distinct from both $S$ and
$V(G) \setminus T$, then we can divide the problem into two smaller subproblems.
When there is no minimum $(S, T)$-separator other than $S$ or $V(G) \setminus
T$, we need to branch on vertices to add to $S$ or $T$. 
For general digraphs, the running time is $n^{2k + O(1)}$: we need
to branch on $O(n)$ vertices from both sides, and the depth of branching is bounded by $k$, as the 
value $d^+(X)$ of the minimum separator $X$ increases at least by one after 
we branch from both sides.  

For $h$-semicomplete digraphs, 
we observe that the number of vertices $v$ such that
$d^+(S \cup \{v\}) \leq k$ is at most $h + 2k + 1$
(see Proposition~\ref{prop:kb-bounded})
and therefore, we need to branch on at most 
$h + 2k + 1$ vertices when extending from $S$.  
Unfortunately, we do not
have a similar bound on the number of vertices to branch on
from the side of $T$.  For example, if $|T| < k$,
then $d^+(V(G) \setminus (T \cup \{v\})) \leq k$ for
every $v \not\in T$ and
therefore we need to branch on every vertex not in $T \cup S \cup \Nout(S)$, 
where $\Nout(S)$ denotes the set of out-neighbors of $S$.

This asymmetry comes from the 
asymmetry inherent in the vertex separation number characterization:
the width of a permutation $\pi$ in $G$ is not equal in general
to the width of a reversal of $\pi$ in $G^{-1}$, the digraph 
obtained from $G$ by reversing all of its edges.
We use separation chains \cite{Pilipczuk12} to
give a symmetric characterization of pathwidth and
formulate a variant of Nagamochi's algorithm which
branches from each side on at most $(h + 2k + 1)$ vertices.
This is how we get the running time stated in Theorem~\ref{thm:alg}.
We remark that a similar result on cutwidth is an immediate 
corollary of the Nagamochi's result, since
we have the desired symmetry in the definition of cutwidth:
the cutwidth of a permutation $\pi$ in $G$ equals the
cutwidth of the reversal of $\pi$ in $G^{-1}$.

The scenario for the combinatorial result in Theorem~\ref{thm:comb}
is rather straightforward.  Given an $h$-semicomplete graph $G$
of pathwidth at least $f(h, k)$, we complete it into a semicomplete
graph $G'$ on $V(G)$, which must have pathwidth at least $f(h, k)$.
We then find an obstacle $T \subseteq V(G)$ in $G'$ 
for small pathwidth, of one of the types defined in \cite{Pilipczuk12}.
Then we consider a random semicomplete subgraph $G''$ of
$G$ and show that $G''$ inherits an obstacle $T'$ from $T$
with high probability such that the existence of $T'$ in $G''$
implies $\pw(G'') \geq k$.
We need to overcome, however, some difficulties in carrying out
this scenario. To be more specific, consider one type of
obstacles, namely {\em degree tangles} \cite{Pilipczuk12}.
An $(l, k)$-degree tangle of $G$ is a vertex set $T$
with $|T| = l$ such that 
$\max_{v \in T} d^+(v) - \min_{v \in T} d^+(v) \leq k$.
In order for a degree tangle $T$ in $G'$ to give rise
to a degree-tangle $T'$ of the random subgraph $G''$, we need
the out-degrees of vertices in $T'$ to ``shrink'' 
almost uniformly. To this end, we wish our sampling to
be such that (1) each vertex $v \in V(G)$ is in $V(G'')$
with a fixed probability $p$ and (2) for each vertex set 
$S \subseteq V(G)$, the intersection $S \cap V(G'')$
has cardinality sharply concentrated around its expectation
$p|S|$.  The following theorem, which may be of independent interest, 
makes this possible: we apply this theorem
to the complement of the underlying graph of $G$ with $d = h$. 

\begin{theorem}
\label{thm:sample-indep}
Let $G$ be an undirected graph on $n$ vertices 
with maximum degree $d$ or smaller.
Let $p = \frac{1}{2d + 1}$.
Then, it is possible to sample a set $I$ of independent vertices of $G$ so that
$\prob(v \in I) = p$ for each $v \in V(G)$ and,   
for each $S \subseteq V(G)$,
we have  
\begin{eqnarray*}
 \prob(|S \cap I| > p|S|+ t) <
 \exp\left(-\frac{t^2}{9|S|}\right)
 \end{eqnarray*}
 and 
 \begin{eqnarray*}
 \prob(|S \cap I| <p|S| - t) <
 \exp\left(-\frac{t^2}{9|S|}\right).
 \end{eqnarray*}
\end{theorem}
Even with this sampling method, it is still not clear if
we can have the desired ``uniform shrinking'' of out-degrees
of the vertices in the degree tangle, since  
if the set $S$ of out-neighbors of a vertex has cardinality
$\Omega(n)$, then the deviation of $|S \cap V(G'')|$ from
its expectation $p|S|$ is necessarily $\Omega(\sqrt{n})$.
To overcome this difficulty, we introduce several
types of obstacles that are robust against random sampling
and show that (1) if $G'$ has an obstacle of a type in \cite{Pilipczuk12}
then it has a robust obstacle and (2) each robust obstacle
in $G'$ indeed gives rise to a strong enough obstacle in $G(V'')$ with
high probability.

A conference version of this paper will appear as \cite{KKT15}.
The rest of this paper is organized as follows.
In Section~\ref{sec:prelim} we define some notation.
In Section~\ref{sec:alg}, we describe our algorithm and prove
Theorem~\ref{thm:alg}.
In Section~\ref{sec:comb}, we prove Theorem~\ref{thm:comb},
assuming Theorem \ref{thm:sample-indep}.
Finally in Section~\ref{sec:sample-indep}, we prove
Theorem~
\ref{thm:sample-indep}.

\section{Notation}
\label{sec:prelim}
Digraphs in this paper are simple: there are no self-loops 
and, between each pair of distinct vertices, there is at most one edge
in each direction. For digraph $G$, $V(G)$ denotes the set of
vertices of $G$ and $E(G) \subseteq V(G) \times V(G)$ the set of edges of $G$.
If $(u, v) \in E(G)$, then $v$ is an {\em out-neighbor} of $u$ and
$u$ is an {\em in-neighbor} of $v$.  
For each $v \in V(G)$, we denote
the set of in-neighbors of $v$ by 
$\NinG(v) = \{u \mid (u, v) \in E(G)\}$ and
write $\NinG[v]$ for $\NinG(v) \cup \{v\}$.
For $U \subseteq V(G)$, we define $\NinG[U] = \bigcup_{v \in U}
\NinG[v]$ and $\NinG(U) = \NinG[U] \setminus U$.
We define the notation for out-neighbors
$\Nout$ similarly.
In this paper, the {\em in-degree} and {\em out-degree} of vertex $v$
in $G$, denoted by $\dinG(v)$ and $\doutG(v)$, respectively, counts
the in-neighbors and out-neighbors rather than the incoming and outgoing edges:
$\dinG(v) = |\NinG(v)|$ and $\doutG(v) = |\NoutG(v)|$;
we also define $\dinG(U) = |\NinG(U)|$ and $\doutG(U) = |\NoutG(U)|$ 
for $U \subseteq V(G)$.
We omit the reference to $G$ from the above notation when it is clear from the context
which digraph is meant. 

\section{Algorithm}
\label{sec:alg}
In this section, we describe the algorithm claimed in Theorem~\ref{thm:alg},
prove its correctness, and analyze its running time.
As suggested in the introduction, our first task is to give a
symmetric characterization of pathwidth to which the Nagamochi's algorithm is
adaptable.
 
Let $G$ be a digraph.
A pair $(A,B)$ of vertex sets of $G$ is a {\em separation} of $G$ if
$A \cup B = V$ and there is no edge from
$A \setminus B$ to $B \setminus A$.
The {\em order} of separation $(A, B)$ is 
$|A \cap B|$. For $S,T \subseteq V$ such that $S \cap T = \emptyset$,
separation $(A, B)$ is an {\em $S$--$T$ separation} if
$S \cap B = \emptyset$ and  
$T \cap A = \emptyset$.
We call an $S$--$T$ separation $(A, B)$ {\em trivial}
if $B = V(G) \setminus S$ or $A = V(G) \setminus T$.

An important role in our algorithm is played by 
a {\em minimum $S$-$T$}
separation, which is defined to be an $S$--$T$ separation 
of the smallest order. 
Note that if a minimum $S$-$T$ separation is trivial, then
it must be either $(\Nout[S],\  V(G) \setminus S)$ or
$(V(G) \setminus T,\  \Nin[T])$. As will be seen later,
we may use non-trivial minimum $S$-$T$ separations to
divide-and-conquer subproblems in our pathwidth computation. 

A sequence of separations 
$((A_0,B_0),(A_1,B_1),\ldots,(A_r,B_r))$ is a {\em separation chain}
if $A_0 \subseteq A_1 \subseteq \ldots \subseteq A_r$ and 
$B_r \subseteq B_{r-1} \subseteq \ldots \subseteq B_0$.
The {\em order} of this separation chain is the maximum order of its member
separations. We use operator $+$ for concatenating sequences of
separations and for appending a separation to a sequence of separations:
for sequences $C$ and $C'$ of separations and a separation $(A, B)$,
$C$ + $C'$ is the concatenation of $C$ and $C'$, $(A, B) + C$ is the
sequence $C$ preceded by $(A, B)$, and $C + (A, B)$ is the sequence
$C$ followed by $(A, B)$.

Let $C = ((A_0,B_0), (A_1, B_2),\ldots,(A_r,B_r))$ be
a separation chain.
We say that $C$ is {\em gapless} if, 
for every $0 < i \leq r$, either $|A_{i} \setminus A_{i - 1}| \leq 1$ or
$|B_{i - 1} \setminus B_{i}| \leq 1$ holds.  Note that this definition
allows a repetition of an identical separation.
We say that $C$ is an {\em $S$--$T$ chain}, 
if $B_0 = V(G) \setminus S$ and $A_r = V(G) \setminus T$, that is,
both ends of $C$ are trivial $S$--$T$ separations.
Note that every separation in an $S$--$T$ chain is an $S$--$T$
separation.

As observed in \cite{Pilipczuk12}, 
\newline (1) if $(X_1,X_2,\ldots,X_r)$ is a path-decomposition of $G$ 
then $((A_0,B_0),(A_1,B_1),$ $\ldots,(A_r,B_r))$,
where $A_i = \bigcup_{j \leq i}X_j$ and $B_i = \bigcup_{i<j}X_j$, is an
$\emptyset$--$\emptyset$ chain in $G$, and 
\newline(2) if $((A_0,B_0),(A_1,B_1),\ldots,(A_r,B_r))$ is 
an $\emptyset$--$\emptyset$ chain in $G$, then 
$(W_1,W_2,\ldots,W_r)$, where $W_i = A_i \cap B_{i-1}$
for $1 \leq i \leq r$, is a path-decomposition of $G$.

These observations lead to the following characterization of
pathwidth by means of gapless separation chains.
\begin{lemma}\label{lem:pathwidth_sc}
Digraph $G$ has a path-decomposition of width $k$ if and only if 
it has a gapless $\emptyset$--$\emptyset$ chain of
order $k$.
\end{lemma}
\begin{proof}
Suppose $G$ has a path-decomposition $(X_1, X_2, \ldots, X_r)$ of
width $k$. We may assume that this path-decomposition is nice:
$X_1 = X_r = \emptyset$ and, for $1 \leq i < r$, either 
$X_{i + 1} = X_i \cup \{v\}$ for some $v \in V(G) \setminus X_i$ or
$X_{i + 1} = X_i \setminus \{v\}$ for some $v \in X_i$.
If we set $A_i = \bigcup_{j \leq i}X_j$ and $B_i = \bigcup_{j > i}X_j$ for 
$0 \leq i \leq r$ as in observation (1), then $((A_0, B_0), (A_2,
B_2),\ldots, (A_r, B_r))$ is a gapless $\emptyset$--$\emptyset$
chain. The order of this separation chain is $\max_{0 \leq i \leq r}
|A_i \cap B_i| = \max_{1 \leq i \leq r - 1}|X_i \cap X_{i+1}| = k$. 
Conversely, suppose a gapless separation chain $((A_0, B_0), (A_1, B_1),\ldots,
(A_r, B_r))$ of order $k$ is given. We set $X_i = A_i  \cap B_{i - 1}$
for $1 \leq i \leq r$. Then, $(X_1, X_2, \ldots, X_r)$ is a
path-decomposition by observation (2).
Since our separation chain is gapless, we have either  $|A_i \setminus A_{i-1}|
\leq 1$ or $|B_{i-1} \setminus B_i| \leq 1$ for $1 \leq i \leq r$.
In the former case, we have $|A_i \cap B_{i-1}| \leq |A_{i-1} \cap B_{i-1}|+1 =
k+1$ and, in the latter case, we have 
$|A_i \cap B_{i-1}| \leq |A_i \cap B_i|+1 = k+1$.  Therefore, the width
of path-decomposition $(X_1, X_2, \ldots, X_r)$ is at most $k$ and
hence $G$ has a path-decomposition of width $k$.
\qed
\end{proof}

We say that a pair $(S, T)$ of vertex sets of $G$ is {\em
$k$-admissible} if $\Nout[S] \cap T = \emptyset$ (and hence $S \cap
\Nin[T] = \emptyset$), $\dout(S) \leq k$, and $\din(T) \leq k$.
It is clear that $(S, T)$ must be $k$-admissible in order
for $G$ to have a gapless $S$--$T$ chain of order at most $k$.
Our algorithm solves the following problem with parameter $k$: given
digraph $G$ and a $k$-admissible pair $(S, T)$, compute a gapless $S$--$T$
chain of order at most $k$ if one exists and otherwise report the non-existence.
The algorithm in Theorem~\ref{thm:alg} applies this algorithm to 
$(S, T) = (\emptyset, \emptyset)$ and, if it returns an $\emptyset$--$\emptyset$
chain of order $k$, converts it to a path-decomposition of width
at most $k$, using the proof of Lemma~\ref{lem:pathwidth_sc}.

The following lemma provides the base case for our
algorithm.
\begin{lemma}\label{lem:base}
If pair $(S,T)$ is $k$-admissible and satisfies 
$|V(G) \setminus (S \cup T)| \leq k + 1$ then 
$G$ has a gapless $S$--$T$ chain of order at most $k$.
\end{lemma}
\begin{proof}
The proof is by induction on $|V(G) \setminus (S \cup T)|$.
The base case is where $V(G) \setminus (S \cup T) = \Nout(S) =
\Nin(T)$. The statement holds in this case, since
the separation $(\Nout[S], \Nin[T])$ alone forms a gapless
$S$--$T$ chain. Since $(S, T)$ is $k$-admissible, 
the order of this separation chain is at most $k$.
Therefore, the base case holds.

Suppose that either $V(G) \setminus (S \cup T) \neq \Nout(S)$ 
or $V(G) \setminus (S \cup T) \neq \Nin(T)$. Consider the
first case: we have some $v \not\in \Nout[S] \cup T$. 
If we set $T' = T \cup \{v\}$, then as $v \not\in \Nout(S)$,
we have $\Nin(T') \subseteq V(G) \setminus (S \cup T \cup \{v\})$ and 
hence we have $|\Nin(T')|
\leq k$. We also have $\Nout[S] \cap T' = \emptyset$
since $v \not\in \Nout(S)$. 
Therefore, $(S, T')$ is $k$-admissible.  
Moreover, we have $|V(G) \setminus (S \cup T')| < |V(G) \setminus (S \cup T)|
\leq k + 1$.  Therefore, we may apply the induction hypothesis to
$(S, T')$ and have a gapless $S$--$T'$ chain $C'$ of order at most $k$.
Let $(A, B)$ be the last separation of $C'$.
Then, since $A = (V(G) \setminus T') \subseteq (V(G) \setminus T)$ and 
$B \supseteq \Nin[T'] \supseteq \Nin[T]$, 
$C = C' + (V(G) \setminus T, \Nin[T])$ is an
$S$--$T$ chain.  Since $C'$ is gapless and $(V(G)
\setminus T) \setminus A = \{v\}$, $C$ is also gapless.
Moreover, since the order of $C'$ is at most $k$ and 
the order of $(V(G) \setminus T, \Nin[T])$ is at most $|\Nin(T)| \leq k$,
the order of $C$ is at most $k$. The second case is similar and symmetric
to the first case.
\qed
\end{proof}

We have two types of recurrences: divide-and-conquer and branching.
For the recurrence of first type, we need the following lemma.
\begin{lemma}\label{lem:submodular}
Suppose $(X,Y)$ is a minimum $S$--$T$ separation.
Then, for each $S$--$T$ separation $(A,B)$, both
$(A \cap X,\  B \cup Y)$ and 
$(A \cup X,\  B \cap Y)$ are 
$S$--$T$ separations and moreover 
neither of their orders exceed that of $(A, B)$.
\end{lemma}
\begin{proof}
Let $A_1 = A \setminus B$, $A_2 = A \cap B$, 
$A_3 = B \setminus A$,
$X_1 = X \setminus Y$, $X_2 = X \cap Y$, and  
$X_3 = Y \setminus X$.
Then, both $(A_1, A_2, A_3)$ and
$(X_1, X_2, X_3)$ partition of $V(G)$.
We have 
\begin{eqnarray*}
(A \cap X) \setminus (B \cup Y) & = & A_1 \cap X_1 \mbox{\ and}\\ 
(B \cup Y) \setminus (A \cap X) & = & 
A_3 \cup X_3
\end{eqnarray*}
and, since there is no edge from $A_1$ to $A_3$ and
no edge from $X_1$ to $X_3$, there is no edge from
$(A \cap X) \setminus (B \cup Y)$ to
$(B \cup Y) \setminus (A \cap X)$.
Therefore, $(A \cap X,\ B \cup Y)$ is a separation and,
similarly, $(A \cup X,\  B \cap Y)$ is a separation.
Since $S \cap B = \emptyset$ and 
$S \cap Y = \emptyset$, we have
$S \cap (B \cup Y) = \emptyset$ and
similarly $(A \cap X) \cap T = \emptyset$.
Therefore, $(A \cap X,\  B \cup Y)$ is an $S$--$T$
separation and, similarly,
$(A \cup X,\  B \cap Y)$ is an $S$--$T$ separation. 

To prove the claim on the orders of these separations,  
we first claim that 
\begin{eqnarray}
\label{eqn:ABXY}
|A \cap B| + |X \cap Y| =
 |(A \cap X) \cap (B \cup Y)|+|(A \cup X) \cap (B \cap Y)|.
\end{eqnarray}
To see this, note that
$A \cap B = A_2$ is partitioned into
$A_2 \cap X_1$, $A_2 \cap X_2$, and $A_2 \cap X_3$;
$X \cap Y = X_2$ is partitioned into
$A_1 \cap X_2$, $A_2 \cap X_2$, and $A_3 \cap X_2$.
On the other hand,
$(A \cap X) \cap (B \cup Y)$ is
partitioned into $A_1 \cap X_2$, 
$A_2 \cap X_2$, and $A_2 \cap X_1$;
$(A \cup X) \cap (B \cap Y)$ is partitioned
into $A_3 \cap X_2$, $A_2 \cap X_2$, and
$A_2 \cap X_3$. Comparing these lists,
we see that both sides of (\ref{eqn:ABXY}) count
the same set of vertices with the same multiplicity.  
Since $(X, Y)$ is a minimum $S$--$T$ separation, we have
$|X \cap Y| \leq |(A \cup X) \cap (B \cap Y)|$ and
hence $|(A \cap X) \cap (B \cup Y)| \leq |A \cap B|$
by (\ref{eqn:ABXY});
similarly we have 
$|(A \cup X) \cap (B \cap Y)| \leq |A \cap B|$.
\qed
\end{proof}

The following lemma, which corresponds to the main lemma in \cite{Naga12}
underlying the algorithm for submodular systems, provides 
the divide-and-conquer type recurrence.
\begin{lemma}\label{lem:divide}
Suppose $G$ has a gapless $S$--$T$ chain of order $k$ and
let $(X,Y)$ be a minimum $S$--$T$ separation of $G$.
Then $G$ has a gapless $S$--$T$ chain of order at most $k$ of
the form $C_1 + (X, Y) + C_2$, where
$C_1$ is a gapless $S$--$(Y \setminus X)$ chain and 
$C_2$ is a gapless $(X \setminus Y)$--$T$ chain.
\end{lemma}
\begin{proof}
Let $C = ((A_0,B_0), (A_1, B_1),\ldots,(A_r,B_r))$ be
an arbitrary gapless $S$--$T$ chain of order at most $k$. 
Recall that $B_0 = V(G) \setminus S$ and $A_r = V(G) \setminus T$
by the definition of $S$--$T$ chains.
Consider the sequence of separations  
$C_1$ consisting of $(A_i \cap X, B_i \cup Y)$ for
$0 \leq i \leq r$. Since we have
$A_{i - 1} \cap X \subseteq A_i \cap X$ and
$B_i \cup Y \subseteq B_{i - 1} \cup Y$ for
$0 < i \leq r$, $C_1$ is a separation chain.
Since $X \cap T = \emptyset$ and $A_r = V(G) \setminus T$,
we have $A_r \cap X = X$.  Therefore,
$C_1$ is an $S$--$(Y \setminus X)$ chain, since 
we have $V(G) \setminus (B_0 \cup Y) = S$ and
$V(G) \setminus (A_r \cap X) = V(G) \setminus X = Y \setminus X$.
Since $C$ is gapless, we have, for each $0 < i \leq r$, 
either $|A_i \setminus A_{i - 1}| \leq 1$ or
$|B_{i - 1} \setminus B_i| \leq 1$.
In the former case, we have 
$|(A_i \cap X) \setminus (A_{i - 1} \cap X)| \leq 1$ and,
in the latter case, we have  
$|(B_{i - 1} \cup Y) \setminus (B_i \cup Y)| \leq 1$.
Therefore, the separation chain $C_1$ is gapless.
By Lemma~\ref{lem:submodular}, the order of $C_1$ is at most $k$.
We similarly construct a gapless $(X \setminus Y)$--$T$ chain $C_2$
of order at most $k$. 

Since the last separation of $C_1$ is $(A_r \cap X, B_r \cup Y)
= (X, B_r \cup Y)$ and the first separation
of $C_2$ is $(A_0 \cup X, B_0 \cap Y) = (A_0 \cup X, Y)$, 
the concatenation $C_1 + (X, Y) + C_2$ is a separation chain
and is moreover gapless.  Since this separation chain
is of order at most $k$ and is an $S$--$T$ chain, 
the lemma holds.
\qed
\end{proof}

We need some preparations before formulating the branching type recurrence.
We say that an $S$--$T$ separation
chain $C=((A_0,B_0),(A_1,B_1),\ldots,(A_r,B_r))$ is {\em nice} if,
for every $0 \leq i < r$, we have $|A_{i+1} \setminus A_i| \leq 1$ and $|B_i
\setminus B_{i+1}|\leq 1$.  We say $C$ is {\em tight} if 
$A_0 = \Nout[S]$ and $B_r = \Nin[T]$.

\begin{lemma}\label{lem:nice_and_tight}
If $G$ has a gapless $S$--$T$ chain of order at most $k$ then it has 
a tight, nice, and gapless $S$--$T$ chain of order at most $k$.
\end{lemma}
\begin{proof}
To each $S$--$T$ chain $C = ((A_0,B_0),(A_1,
B_1),\ldots,(A_r,B_r))$, we assign a non-negative integer $\delta(C)$ by
\begin{eqnarray*}
	\delta(C) & = & |A_0 \setminus \Nout[S]| + |B_0 \setminus \Nin[T]| \\
	&& + \sum_{0 \leq i < r}(\max\{0,|A_{i+1} \setminus A_i|-1 \} + \max\{0,|B_i
	\setminus B_{i+1}|-1 \}).
\end{eqnarray*}
Choose a gapless $S$--$T$ chain $C = ((A_0,B_0),(A_1,
B_1),\ldots,(A_r,B_r))$ to minimize $\delta(C)$
subject to being of order at most $k$.
If $\delta(C)=0$ then $C$ is tight and nice and we
are done.
For contradiction, suppose $\delta(C) > 0$.
We first consider the case where there is some vertex $v \in A_0 \setminus
\Nout[S]$.  Let $C'$ be obtained from $C$ by adding
separation $(A_0 \setminus \{v\},B_0)$ before $C$. 
Then, $C'$ is a gapless $S$--$T$ chain.  The order of separation
$(A_0 \setminus \{v\},B_0)$ is smaller 
than that of $(A_0,B_0)$ and hence the order of $C'$ is at most $k$.
This contradicts the choice of $C$ since $\delta(C') = \delta(C)-1$.
We have similarly a contradiction if there is some $v \in B_r \setminus
\Nin[T]$.
Suppose finally that $|A_{i+1} \setminus A_i| \geq 2$ for some 
$0 \leq i < r$. Let $v$ and $v'$ be two distinct vertices in $A_{i+1}
\setminus A_i$.
Now, since $C$ is gapless, this assumption implies that  
$|B_i \setminus B_{i+1}| \leq 1$.
As neither $v$ nor $v'$ is in $A_i$ and hence both are in $B_i$, 
it follows that 
$|A_{i + 1} \cap B_{i + 1}| \geq |A_i \cap B_i| + 1$. 
Since $|(A_i \cup \{v\}) \cap B_i| = |A_i \cap B_i| + 1$, 
the order of separation $(A_i \cup \{v\}, B_i)$ is no greater than that of
$(A_{i+1},B_{i+1})$ and hence is at most $k$.
Therefore the $S$--$T$ chain $C'$ that is obtained from $C$ by
placing $(A_i \cup \{v \}, B_i)$ between $(A_i, B_i)$ and $(A_{i + 1}, B_{i + 1})$
is gapless and of order at most $k$.
We have 
\begin{eqnarray*}
		\delta(C')	& = & \delta(C) 
		- \max\{0,|A_{i+1} \setminus A_i|-1\} 
		- \max\{0,|B_i \setminus B_{i+1}|-1\}\\
		& & + \max\{0,|\{v\}|-1\} 
		+ \max\{0,|B_i \setminus B_i|-1\}\\
		& & + \max\{0,|A_{i+1} \setminus (A_i \cup \{v \})|-1\} 
		+ \max\{0,|B_i \setminus B_{i+1}|-1\}\\
		& = & \delta(C) - \max\{0,|A_{i+1} \setminus A_i|-1]\} 
		+ \max\{0,|A_{i+1} \setminus (A_i \cup \{v \})|-1\}.
		\end{eqnarray*}							
Since $|A_{i+1} \setminus A_i| > |A_{i+1} \setminus (A_i \cup \{v \})|> 0$, it
follows that $\delta(C') \leq \delta(C) - 1$, a contradiction.
We similarly obtain a contradiction from the case 
$|B_i \setminus B_{i+1}| \geq 2$ as well.
\qed
\end{proof}	

The following lemma provides our branching type recurrence.
\begin{lemma}
\label{lem:sc-add}
Suppose $G$ has a gapless $S$--$T$ chain of order at most $k$ and
suppose that $|V(G) \setminus (S \cup T)| \geq k + 2$ holds.
Then, there are a gapless $S$--$T$ chain
$((A_0, B_0), \ldots, (A_r, B_r))$ of order at most $k$ and a pair of distinct
vertices $u \in V(G) \setminus (S \cup \Nin[T])$ and 
$v \in V(G) \setminus (T \cup \Nout[S])$ such that the following holds:
\begin{enumerate}
\item $((A_1, B_1), \ldots, (A_r, B_r))$ is an 
$(S \cup \{u\})$--$T$ chain,
\item $((A_0, B_0), \ldots, (A_{r-1}, B_{r-1}))$ is an 
$S$--$(T \cup \{v\})$ chain, and 
\item $((A_1, B_1), \ldots, (A_{r-1}, B_{r-1}))$ is an 
$(S \cup \{u\})$--$(T \cup \{v\})$ chain.
\end{enumerate}
\end{lemma}
\begin{proof}
Suppose $G$ has a gapless $S$--$T$ chain of order at most $k$.
By Lemma~\ref{lem:nice_and_tight}, $G$ has a gapless  
$S$--$T$ chain $C=((A_0,B_0), (A_1, B_1),\ldots,(A_r,B_r))$ of
order at most $k$ that is tight and nice.
Since $C$ is tight, we have $B_r = \Nin[T]$.
We also have $B_0 = V(G) \setminus S$ from the definition of an
$S$--$T$ chain.  Therefore, 
$B_0 \setminus B_r = V(G) \setminus (S \cup \Nin[T])$ and
this set contains at least two vertices as we are assuming
$|V(G) \setminus (S \cup T)| \geq k + 2$.
Similarly $A_r \setminus A_0 = V(G) \setminus (\Nout[S] \cup T)$ 
has at least two vertices. 
Let $i_1$ denote the smallest $i$ 
such that $0 < i \leq r$ and $|B_{i-1} \setminus B_i| = 1$
and $i_2$ the largest $i$ such that
$0 \leq  i < r$ and $|A_{i + 1} \setminus A_i| = 1$
Since $C$ is nice, the choice of $i_1$ and $i_2$ implies that
$B_i = B_0$ for $0 \leq i < i_1$ and
$A_i = A_r$ for $i_2 < i \leq r$.
Let $u$ be the unique vertex in $B_{i_1 - 1} \setminus B_{i_1}$ and
$v$ the unique vertex in $A_{i_2 + 1} \setminus A_{i_2}$.
We must have $i_1 \leq i_2$, since otherwise
$A_{i_1} \cap B_{i_1} = A_r \cap (B_0 \setminus \{u\})
= (V(G) \setminus T) \cap (V(G) \setminus (S \cup \{u\}))
= V(G) \setminus (S \cup T \cup \{u\})$ 
and hence $|V(G) \setminus (S \cup T)| \leq |A_{i_1} \cap B_{i_1}| + 1 \leq k +
1$, contradicting our assumption.

Since $u \not\in B_{i_1}$ and $v \not\in A_{i_1} \subseteq A_{i_2}$,
we must have $u \neq v$.  Let $C'$ be the separation
chain $((A_{i_1}, B_{i_1}), \ldots, (A_{i_2}, B_{i_2}))$
Then, $(A_0, B_0) + C' + (A_r, B_r)$ is a $S$--$T$ chain since
$B_0 = V(G) \setminus S$ and $A_r = V(G) \setminus T$,
it is gapless since $|B_0 \setminus B_{i_1}| = 1$ and 
$|A_r \setminus A_{i_2}| = 1$, and it is clearly of degree at most $k$.
Since $B_{i_1} = V(G) \setminus (S \cup \{u\})$ and  
$A_{i_2} = V(G) \setminus (T \cup \{v\})$, 
$(A_0, B_0) + C'$ is an $S$--$(T \cup \{v\})$ chain and
$C' + (A_r, B_r)$ is an $(S \cup \{u\})$--$T$ chain.
Therefore, the separation chain
$(A_0, B_0) + C' + (A_r, B_r)$ qualifies as the $S$--$T$ chain claimed in the lemma.
\qed
\end{proof}

Given these recurrences and the base case above, our algorithm is
straightforward. Suppose we are given a $k$-admissible pair $(S, T)$.
If $|V(G) \setminus (S \cup T)| \leq k + 1$ holds then we apply
Lemma~\ref{lem:base} and return the gapless $S$--$T$ chain 
it provides. Suppose otherwise. We test if there is
a minimum $S$--$T$ separation that is non-trivial: a minimum $S$--$T$ separation
$(X, Y)$ that is not equal to either $(\Nout[S],\  V(G) \setminus S)$ or
$(V(G) \setminus T,\  \Nin[T])$.
If we find one, we apply Lemma~\ref{lem:divide} and recurse on
subproblems $(S,\ Y \setminus X)$ and $(X \setminus Y,\ T)$.
If either of the recursive calls returns a negative answer, we
return a negative answer.  Otherwise, we concatenate the
solutions from the subproblems as prescribed in Lemma~\ref{lem:divide}
and return the result.  Finally suppose that
there is no minimum $S$--$T$ separation that is non-trivial. 
If $(\Nout[S], V(G) \setminus S)$ is the only minimum $S$--$T$ separation,
then we recurse on $(S \cup \{v\},\  T)$ for every $v \in V(G) \setminus (S \cup
T)$ such that $(S \cup \{v\},\  T)$ is $k$-admissible. 
If $(V(G) \setminus T, \Nin[T])$ is the only
minimum $S$--$T$-separation, then we similarly branch from $T$.
If both $(\Nout[S],\  V(G) \setminus S)$ and $(V(G) \setminus T,\  \Nin[T])$ are
the minimum $S$--$T$ separations, then we branch from both sides.
In either case, if any of the recursive call returns a gapless separation chain
of order at most $k$, 
we trivially extend the chain into a gapless $S$--$T$ separation of order at
most $k$ and return this chain. Otherwise, that is, if all the recursive calls return
negative answers, we return a negative answer.

The correctness of this algorithm is proved by a straightforward induction 
for which the above Lemmas provide the base case and the induction steps.

We analyze the running time of the algorithm.
The following observation extends the one in \cite{Pilipczuk12} that 
the number of vertices of out-degree at most $k$ in a semicomplete
digraph is at most $2k + 1$.

\begin{proposition}
\label{prop:kb-bounded}
Let $G$ be an $h$-semicomplete digraph and let $U \subseteq V(G)$.
Then the number of vertices $v \in V(G) \setminus U$ such that
$\dout(U \cup \{v\}) \leq k$ is at most $h + 2k + 1$ for every $k > 0$.
The similar statement with the out-degree replaced by the in-degree also holds.
\end{proposition} 
\begin{proof}
Fix $U$, let $X \subset V(G) \setminus U$ be arbitrary,
and set $|X| = b$.
By the definition of $h$-semicomplete digraphs, 
$G[X]$ contains at least $b(b - h - 1)/2$ edges and
hence the average out-degree of vertices in $G[X]$ is at least 
$(b - h - 1) / 2$. For each $v \in X$, $\NoutG(U \cup \{v\})$
contains $N^+_{G[X]}(v)$ and hence if $b > h + 2k + 1$ then
there is at least one $v \in X$ such that $|\NoutG(U \cup \{v\})| > k$.
This proves the first statement. The second statement is immediate by
symmetry.
\qed
\end{proof}

Thus, the number of vertices to branch on from each side in the above algorithm
is bounded by $h + 2k + 1$. 

To measure the ``size" of the problem instance $(S, T)$, we
introduce the following two functions.
Let $\gamma(S, T)$ denote the order of the minimum $S$--$T$ separation.
Let $\mu(S, T)$ be defined by 
\begin{eqnarray*}
\mu(S,T) & = & 2|V(G) \setminus (\Nout[S]\cup \Nin[T])| + 
|\Nout(S) \Delta \Nin(T)|,
		\end{eqnarray*}
where $X \Delta Y$ is the symmetric difference between $X$ and $Y$.
	
\begin{lemma}\label{lem:rec-bound}
Let $(X,Y)$ be a minimum $S$--$T$ separation.  Then, we have
\[	\mu(S,\ Y \setminus X) + \mu(X \setminus Y,\ T) = \mu(S,T).	\]
\end{lemma}
\begin{proof}
Since $(X,Y)$ is a minimum $S$--$T$ separation, we have  
$\Nout(X \setminus Y) = \Nin(Y \setminus X) = X \cap Y$ and 
hence $\Nout[X \setminus Y] = X$ and 
$\Nin[Y \setminus X] = Y$.
We define pairwise disjoint vertex sets $C_0$, $C_1$, and $C_2$ by
\begin{eqnarray*}
C_0 & = & X \cap Y \setminus (\Nout(S) \cup \Nin(T))\\
C_1 & = & X \cap Y \cap (\Nout(S) \setminus \Nin(T))\\
C_2 & = & X \cap Y \cap (\Nin(T) \setminus \Nout(S)).
\end{eqnarray*}
Then, noting that $(X \cap Y)\setminus \Nin(T) = C_0 \cup C_1$
and that $\Nin(T) \setminus (X \cap Y) = \Nin(T) \setminus X$ since
$\Nin(T) \cap (X \setminus Y) = \emptyset$, we have
\begin{eqnarray*}
\mu(X \setminus Y,\  T) & = & 
2|V(G) \setminus (\Nout[X \setminus Y] \cup \Nin[T])| + |
\Nout(X \setminus Y) \Delta \Nin(T)|\\
& = & 2|V(G) \setminus (X \cup \Nin[T])| + 
|(X \cap Y) \Delta \Nin(T)| \\
& = & 2|V(G) \setminus (X \cup \Nin[T])|
	+ |C_0| + |C_1| + |\Nin(T) \setminus X|.
\end{eqnarray*}
Similarly, we have 
\begin{eqnarray*}
\mu(S,\  Y \setminus T) & = &  
2|V(G) \setminus (\Nout[S] \cup Y)| +
|\Nout(S) \Delta (X \cap Y)| \\
& = & 2|V(G) \setminus (\Nout[S] \cup Y)|
+ |C_0| + |C_2| + |\Nout(S) \setminus Y|.
\end{eqnarray*}
Moreover, we have 
\begin{eqnarray*}
|V(G) \setminus (\Nout[S] \cup \Nin[T])| & = & 
|V(G) \setminus (Y \cup \Nout[S])|
+ |V(G) \setminus (X \cup \Nin[T])|+ |C_0|
\end{eqnarray*}
and 
\begin{eqnarray*}
|\Nout(S) \Delta \Nin(T)| & = & |C_1| + |C_2| 
	+ |\Nout(S) \setminus Y| + |\Nin(T) \setminus X|.
\end{eqnarray*}
Therefore, we have
\begin{eqnarray*}
\mu(S,T) &=& 2|V(G) \setminus (\Nout[S] \cup \Nin[T])|+ |\Nout(S)
	\Delta \Nin(T)|	\\
	& = & 
2|V(G) \setminus (Y \cup \Nout[S])| +
2|(V(G) \setminus (X \cup \Nin[T])|\\
  & & + 2|C_0|
	+ |C_1| + |C_2|
	+ |\Nout(S) \setminus Y| + |\Nin(T) \setminus X|\\
	& = & \mu(S,\  Y \setminus X) + \mu(X \setminus Y,\  T)
\end{eqnarray*}
as claimed in the lemma.
\qed
\end{proof}

\begin{lemma}\label{lem:mu-split}
Let $(X, Y)$ be a non-trivial $S$--$T$ separation:
$X \setminus Y \neq S$ and $Y \setminus X \neq T$.
Then, we have $\mu(S,\ Y \setminus X) \ge 1$ and 
$\mu(X \setminus Y,\  T) \ge 1$.
\end{lemma}
\begin{proof}
Due to the symmetry it suffices to prove the first inequality.		
From the assumption, there is some vertex $v \in (X \setminus Y) \setminus
S$. Since $\Nin[Y \setminus X] \subseteq Y$, we have 
$v \not\in \Nin[Y \setminus X]$.
If $v \in \Nout(S)$ then $v \in \Nout(S) \Delta \Nin(Y \setminus X)$ and 
otherwise $v \in V(G) \setminus (\Nout[S]\cup \Nin[Y \setminus X])$.
Therefore, in either case, we have  
\begin{eqnarray*}
\mu(S,\ Y \setminus X) & = & 2|V(G) \setminus (\Nout[S]\cup \Nin[Y \setminus
X])|
\\
& & +|\Nout(S) \Delta \Nin(Y \setminus X)|\\
& \ge & 1.
\end{eqnarray*}
\qed
\end{proof}

Let $R(S,T)$ denote the number of problem instances recursively considered
when we solve the instance $(S, T)$, not counting the instances in the base
case, 
but counting the instance $(S, T)$ itself unless it is in the base case.
Let $\mu'(S, T) = \max\{0, 2\mu(S, T) - 1\}$.

\begin{lemma}
\label{lem:num_rec}
Let $G$ be an $h$-semicomplete digraph and $k$ a positive integer.
Then, for each $k$-admissible pair $(S, T)$, 	
we have 
\begin{eqnarray*}
\label{eqn:R}
	R(S,T) \leq \mu'(S, T) \cdot (h + 2k + 1)^{2(k-\gamma(S,T))}
\end{eqnarray*}
\end{lemma}
\begin{proof}
The proof is by induction on the structure of recursive calls.		
If instance $(S, T)$ belongs to the base case 
$|V(G) \setminus (S \cup T)| \leq k + 1$, then $R(S, T) = 0$
by definition and inequality (\ref{eqn:R}) trivially holds.
Note that if $\mu(S, T) =
0$ then $V(G) \setminus (S \cup T) = \Nin(S) = \Nout(T)$ and
$(S, T)$ belongs to the base case.
We next consider the case where, in processing the instance  
$(S, T)$, the ``divide-and-conquer" recurrence is applied
and instances $(S,T')$ and $(S',T)$ are recursed on.
We have a non-trivial minimum separation $(X, Y)$ of $(S, T)$
such that $S' = X \setminus Y$ and $T' = Y \setminus X$.
By Lemma~\ref{lem:rec-bound}, we have 
$\mu(S, T) = \mu(S, T') + \mu(S', T)$.
Moreover, by Lemma~\ref{lem:mu-split},
we have $\mu(S, T') \ge 1$ and $\mu(S', T) \ge 1$.
Therefore, we have 
\begin{eqnarray*}
	\mu'(S, T) & = & 2\mu(S, T) - 1 \\
	& = & (2\mu(S, T') - 1) + (2\mu(S', T) - 1) + 1\\
	& = & \mu'(S, T') + \mu'(S', T) + 1.
\end{eqnarray*}
Moreover, we have $\gamma(S, T') \geq \gamma(S, T)$ since
every $S$--$T'$ separation is a $S$--$T$ separation 
and similarly $\gamma(S', T) \geq \gamma(S, T)$.
Applying the induction hypothesis to the instances
$(S, T')$ and $(S', T)$, we have 
\begin{eqnarray*}
	R(S,T) & =	&	1+R(S,T')+R(S',T) \\
	& \leq & 1 + (\mu'(S, T') + \mu'(S', T)) \cdot b^{2(k-\gamma(S,T))}\\
	& \leq & 1 + (\mu'(S, T) - 1) \cdot b^{2(k-\gamma(S,T))}\\
	& \leq & \mu'(S, T) \cdot b^{2(k-\gamma(S,T))},\\
\end{eqnarray*}
where $b = h + 2k + 1$, 
that is, inequality (\ref{eqn:R}).
We next consider the case where the branching recurrence is applied.
We have three cases to consider: (1) $(\Nout[S], V(G) \setminus S)$
and $(V(G) \setminus T, \Nin[T])$ are the only minimum $S$--$T$ separators,
(2) $(\Nout[S], V(G) \setminus S)$ is the only minimum $S$--$T$ separator,
and (3) $(V(G) \setminus T, \Nin[T])$ is the only minimum $S$--$T$ separator.
First consider case (1).  In this case, 
for each pair of vertices $u \in V \setminus (S \cup \Nin[T])$
and $v \in V \setminus (\Nout[S] \cup T)$ such that
the pair $(S \cup \{u\}, T \cup \{v\})$ is $k$-admissible, 
the instance $(S \cup \{u\}, T \cup \{v\})$ is recursed on.
By Proposition~\ref{prop:kb-bounded}, the number of such pair
is at most $b^2 = (h + 2k + 1)^2$.
For each pair of $u$ and $v$, we have by the induction hypothesis 
\begin{eqnarray*}
	R(S \cup \{u\}, T \cup \{v\})  
	\leq  \mu'(S \cup \{u\}, T \cup \{v\}) \cdot 
	b^{2(k-\gamma(S \cup \{u\}, S \cup
	\{v\}))}.
	\end{eqnarray*}
Since no $(S \cup \{u\}$--$(T \cup \{v\})$ separation is a minimum $S$--$T$
separation from the assumption of this case,  	
we have $\gamma(S \cup \{u\}, T \cup \{v\}) > \gamma(S, T)$.
Moreover, since $\mu(S \cup \{u\}, T \cup \{v\}) < \mu(S, T)$ and 
$\mu(S, T) > 0$, we have 
$\mu'(S \cup \{u\}, T \cup \{v\}) < \mu'(S, T)$. 
Therefore, we have  
\begin{eqnarray*}
	R(S, T) &\leq & 1 + \sum_{u, v}R(S \cup \{u\}, T \cup \{v\})\\
	&\leq & 1 + b^2 \cdot (\mu'(S, T) - 1) \cdot b^{2(k-\gamma(S,T) - 1)} \\
	&\leq &  \mu'(S, T)\cdot b^{2(k-\gamma(S,T))}, \\
\end{eqnarray*}
that is, inequality (\ref{eqn:R}).
Cases (2) and (3) are similar and somewhat simpler.
\qed
\end{proof}

The time for processing each pair $(S, T)$ excluding the time
consumed by subsequent recursive calls is dominated by the
time for finding minimum $S$--$T$ separation and for deciding
if there is a minimum $S$--$T$ separation that is not trivial.
This can be done in $n^{O(1)}$ time by the repeated use
of a standard augmenting path algorithm for a minimum $S$--$T$ cut.
Since $\mu'(\emptyset, \emptyset) = O(n)$, 
we have the running time claimed in Theorem~\ref{thm:alg}.

\section{Tame obstacles survive random sampling: proof of
Theorem~\ref{thm:comb}}
\label{sec:comb}
We prove Theorem~\ref{thm:comb} in this section.

Let $G$ be a semicomplete digraph with $n$ vertices.
For $0 \leq d \leq n$, let
$\Vinle{d}(G)$, $\Vinge{d}(G)$, $\Voutle{d}(G)$, and $\Voutge{d}(G)$ 
denote the set of vertices $v$ with
$\dinG(v) \leq d$, $\dinG(v) \geq d$, $\doutG(v) \leq d$,
and $\doutG(v) \geq d$, respectively.
We omit the reference to $G$ and write $\Vinle{d}$ etc. 
when $G$ is clear from the context.

\begin{proposition}
For every $0 \leq d < n$, we have
$\Voutle{d} \subseteq \Vinge{n - d - 1}$
and  $\Vinle{d} \subseteq \Voutge{n - d - 1}$.
\end{proposition}

\begin{definition}
\label{def:tangles}
\cite{Pilipczuk12}
Let $G$ be a semicomplete digraph and
let $d \geq 0$, $l > 0$ and $k > 0$ be integers. 
A {\em $(d, l, k)$-degree tangle} 
of $G$ is a vertex set $T \subseteq \Voutge{d} \cap \Voutle{d + k}$
with $|T| =  l$.
An {\em $(d, l, k)$-matching tangle} 
of $G$ is a pair of vertex sets $(T_1, T_2)$
with $|T_1| =|T_2| = l$ such that:
\begin{enumerate}
  \item 
  $T_1 \subseteq \Voutle{d}$, 
  $T_2 \subseteq \Voutge{d + k + 1}$, and
  \item there is some bijection $\phi:
T_1 \rightarrow T_2$ such that 
$(v, \phi(v)) \in E(G)$ for every $v \in T_1$.
\end{enumerate}
We will often refer to a $(d, l, k)$-degree (-matching) tangle as an $(l,
k)$-degree (-matching) tangle without specifying $d$.
\end{definition}

\begin{lemma}
\label{lem:degree-interval}
Let $G$ be a semicomplete digraph on $n$ vertices.
Then, 
for each pair $d_1$ and $d_2$ of non-negative integers
such that $d_1 + d_2 < n$,
we have $|\Voutge{d_1} \cap \Vinge{d_2}| \leq n - (d_1 + d_2) + 2\pw(G)$.
\end{lemma}
\begin{proof}
Fix an optimal nice path-decomposition $X_0, X_1, \ldots, X_{2n}$
of $G$, where $n = |V(G)|$.
We say that vertex $v$ is {\em introduced at $i$} if
$X_{i} \setminus X_{i - 1} = \{v\}$ and {\em forgotten at $i$} if
$X_{i - 1} \setminus X_i = \{v\}$.
Let $i_0$ denote the smallest index $i$ such that
a vertex in $\Voutge{d_1} \cap \Vinge{d_2}$ is forgotten at $i + 1$; 
we let $v_0$ denote this forgotten vertex.
Similarly, let $i_1$ be the largest index $i$ such that 
a vertex in $\Voutge{d_1} \cap \Vinge{d_2}$ is introduced at $i$; 
we let $v_1$ denote this vertex.
If $i_0 \geq i_1$ then $\Voutge{d_1} \cap \Vinge{d_2} \subseteq X_{i_0}$ and
hence $|\Voutge{d_1} \cap \Vinge{d_2}| \leq \pw(G) + 1$; we are done.
So suppose that $i_0 < i_1$.
Let $Y_0 = \bigcup_{j \leq i_0} X_{j}$ and
$Y_1 = \bigcup_{j \geq i_1} X_{j}$.
Since $\Nout[v_0] \subseteq Y_0$, by the definition of path-decompositions, 
and $\dout(v_0) \geq d_1$, we have $|Y_0| \geq d_1 + 1$.
Similarly, since $\Nin[v_1] \subseteq Y_1$ and 
$\din(v_1) \geq d_2$ we have $|Y_1| \geq d_2 + 1$.
Let $Z$ be the set of vertices in $\Voutge{d_1} \cap \Vinge{d_2}$ 
that are introduced at some $i > i_0$ and 
forgotten at some $i' < i_1$.
Then, each vertex in $(\Voutge{d_1} \cap \Vinge{d_2}) \setminus Z$ 
must be in $X_{i_0}$ if it is introduced at some $i \leq {i_0}$
and in $X_{i_1}$ if it is forgotten at some $i > {i_1}$.
As $Y_0 \cup Y_1 \subseteq V(G) \setminus Z$, 
we have 
\begin{eqnarray*}
|Y_0 \cup Y_1| & \leq & 
n - |\Voutge{d_1} \cap \Vinge{d_2}| + 
|(\Voutge{d_1} \cap \Vinge{d_2}) \setminus Z|\\
& \leq & n - |\Voutge{d_1} \cap \Vinge{d_2}| +
|X_{i_0} \cup X_{i_1}|.
\end{eqnarray*}
We have $Y_0 \cap Y_1 = X_{i_0} \cap X_{i_1}$ from
the definition of a path-decomposition and hence
$|Y_0| + |Y_1| \leq n - |\Voutge{d_1} \cap \Vinge{d_2}| + |X_{i_0}| +
|X_{i_1}|$.
Combining with the bounds on $|Y_0|$ and $|Y_1|$
above, we have 
\begin{eqnarray*}
|\Voutge{d_1} \cap \Vinge{d_2}| & \leq & n - (d_1 +
1) - (d_2 + 1) + |X_{i_0}| + |X_{i_1}|\\
& \leq & n - (d_1 + d_2) + |X_{i_0}| + |X_{i_1}| - 2\\
& \leq & n - (d_1 + d_2) + 2\pw(G).
\end{eqnarray*}
\qed
\end{proof}

\begin{corollary}
\label{cor:degree-tangle}
If $G$ has an $(l, k)$-degree tangle then
$\pw(G) \geq (l - k - 1)/2$. 
\end{corollary}
\begin{proof}
Let $T$ be a $(l, k)$-degree tangle.
Then, $T \subseteq \Voutge{d} \cap \Voutle{d + k}
\subseteq \Voutge{d} \cap \Vinge{n - (d + k) -1}$
for some $d$ and hence $l \leq n - (n - k - 1) + 2\pw(G) = k + 1 + 2\pw(G)$
by Lemma~\ref{lem:degree-interval}. The corollary follows.
\qed
\end{proof}
\begin{remark}
The lemma in \cite{Pilipczuk12} states that if
$G$ has a $(5k+2, k)$-degree tangle then $\pw(G) > k$.
The above corollary implies a slightly stronger statement that if $G$ has
a $(3k+2, k)$-degree tangle then $\pw(G) > k$. 
\end{remark}

The following lemma generalizes the analysis of
on matching tangles in \cite{Pilipczuk12}.
We need this generalization when we introduce another
obstacle for small pathwidth. 
\begin{lemma}
\label{lem:disjoint}
Let $G$ be a semicomplete digraph on $n$ vertices and
let $l$, $k$, $d$ positive integers.
Suppose $G$ has a set of $l$ pairwise vertex-disjoint
directed paths from $\Voutle{d}$ to 
$\Voutge{d + k}$. Then, $\pw(G) \geq \min\{l, k\}$.
\end{lemma}
\begin{proof}
Let $Q$ be a set of $l$ pairwise vertex-disjoint
directed paths from $\Voutle{d}$ to 
$\Voutge{d + k}$. 
We assume $\pw(G) \leq k - 1$ and show that
$\pw(G) \geq l$.
Let $X_0$, \ldots, $X_{2n}$ be a nice path-decomposition of $G$ of 
optimal width (which is $k - 1$ or smaller).
Let $A_i = \bigcup_{j \leq i} X_j$ and $B_i = \bigcup_{j > i} X_j$ for
$0 \leq i < 2n$.  Since $|A_i| + |B_i| 
= n + |A_i \cap B_i| \leq n + k - 1$
holds for $0 \leq i < 2n$, there is some $i$ such that 
$|A_i| \leq d + k$
and $|B_i| \leq n - d - 1$. Fix such $i$. 
For each $v \not\in A_i$, $\Nin[v] \subseteq B_i$ and
hence $\din(v) \leq n - d - 2$.
Therefore, we have $\Voutle{d} \subseteq \Vinge{n - d - 1} \subseteq A_i$.   
Similarly, 
for each $v \not\in B_i$, $\Nout[v] \subseteq A_i$ and
hence $\dout(v) \leq d + k - 1$. Therefore we have 
$\Voutge{d + k} \subseteq B_i$.
Therefore, each path in $Q$ from $\Voutle{d}$ to 
$\Voutge{d + k}$ is from $A_i$ to $B_i$ and 
must have at least one vertex in $A_i \cap B_i$ since 
$(A_i, B_i)$ is a separation and hence there is no edge 
from $A_i \setminus B_i$ to $B_i \setminus A_i$.
As the $l$ paths in $Q$  
are pairwise vertex-disjoint, 
we have $l \leq |A_i \cap B_i| \leq \pw(G)$.
\qed
\end{proof}

\begin{corollary}\cite{Pilipczuk12}
If a semicomplete digraph $G$ has a 
$(l, k)$-matching tangle, then
$\pw(G) \geq \min\{l, k + 1\}$.
\end{corollary}
\begin{proof}
Let $(T_1, T_2)$ be a $(l, k)$-matching tangle
and let $d$ be such that
$T_1 \subseteq \Voutle{d}$ and 
$T_2 \subseteq \Voutge{d + k + 1}$.
Apply Lemma~\ref{lem:disjoint} to
the set of $l$ vertex-disjoint paths from
$\Voutle{d}$ to $\Voutge{d + k + 1}$
provided by the matching edges.
\qed
\end{proof}

We follow the scenario described in the introduction.
Given an $h$-semicomplete
digraph $G$ of pathwidth at least $f(h, k)$, we complete it into a semicomplete
digraph $G'$ on $V(G)$, in which we find a large obstacle, say a degree tangle
$T$. Then, we apply Theorem~\ref{thm:sample-indep} to obtain a
random independent set $I$ of the complement of the underlying graph of $G$.
We hope that $T \cap I$ is a tangle of $G[I]$
that is strong enough to conclude $\pw(G[I]) \geq k$.
For this to happen, we need to have the out-degrees $|N^+_{G'}(v) \cap I|$
of $v$, for $v \in T \cap I$, to be close to each other.

As observed in \cite{Pilipczuk12}, the optimal vertex separation sequence
lists the vertices roughly in the order of increasing out-degrees and 
therefore each vertex has most vertices of smaller degree as its
out-neighbors, except for some exceptions.  
The following notion of the wildness of vertices measures how
exceptional a vertex is.
\begin{definition}
\label{def:wildness}
For each vertex $v \in G$, we define the {\em wildness}
$\wild(v)$ of $v$ by 
\begin{eqnarray*}
\wild(v) = |\Voutle{\dout(v)} \setminus \Nout(v)|.
\end{eqnarray*}
\end{definition}

\begin{lemma}
\label{lem:wild}
Let $G$ be semicomplete and $v$ an arbitrary vertex of $G$. 
Then, for each integer $w \geq 0$, 
we have 
\begin{eqnarray*}
|\Voutle{\dout(v) - w} \cap \Nin(v)| \geq \wild(v) - w - 2\pw(G) - 1
\end{eqnarray*}
and 
\begin{eqnarray*}
|\Voutge{\dout(v) + w} \cap \Nout(v)| \geq \wild(v) - w - 2\pw(G).
\end{eqnarray*}
\end{lemma}
\begin{proof}
For the first inequality, first observe that
\begin{eqnarray*}
|\Voutle{\dout(v)} \cap \Nin(v)| 
& \geq & |\Voutle{\dout(v)} \setminus \Nout(v)| - 1\\
& = & \wild(v) - 1, 
\end{eqnarray*}
since each vertex not in $\Nin(v)$ must be in 
$\Nout(v) \cup \{v\}$.
Since 
\begin{eqnarray*}
|\Voutle{\dout(v)} \setminus \Voutle{\dout(v) - w}|
& \leq & |\Vinge{n - \dout(v) - 1} \cap \Voutge{\dout(v) - w + 1}| \\
& \leq & w + 2\pw(G)
n\end{eqnarray*}
by Lemma~\ref{lem:degree-interval} (or trivially holding when
$w = 0$ and hence Lemma~\ref{lem:degree-interval} is not applicable), 
we obtain the first inequality.
 
For the second inequality, we have 
$|\Voutle{\dout(v) + w - 1}| \leq  |\Vinge{n - (\dout(v) + w)}| \leq  
\dout(v) + w + 2\pw(G)$ by Lemma~\ref{lem:degree-interval} and hence
\begin{eqnarray*} 
|\Voutle{\dout(v) + w - 1} \cap \Nout(v)| &\leq & \dout(v) + w + 2\pw(G) -
|\Voutle{\dout(v) + w - 1} \setminus \Nout(v)| \\
& \leq & 
\dout(v) + w + 2\pw(G) -
|\Voutle{\dout(v)} \setminus \Nout(v)| \\
& = & \dout(v) + w + 2\pw(G) - \wild(v).
\end{eqnarray*}
Therefore, of the $\dout(v)$ vertices in 
$\Nout(v)$, at least $\wild(v) - w - 2\pw(G)$ must
belong to $\Voutge{\dout(v) + w}$.
\qed
\end{proof}

If the vertices of a degree-tangle $T$ have small wildness,
then most of their out-neighbors are shared and we
may expect that their degrees in the sampled subgraph $G[I]$ 
will be close to each other. We call such a degree-tangle tame.

\begin{definition}
\label{def:tame}
We say that an $(l, w)$-degree tangle $T$ of $G$ is {\em
tame} (relative to the parameters $l$ and $w$), 
if $\wild(v) \leq 3l + w + 2\pw(G)$ for each $v \in T$.
\end{definition}

A degree-tangle is not necessarily tame, but 
a large number of wild vertices in a degree-tangle are themselves
an evidence of large pathwidth. We capture this fact by another
type of obstacles we call spiders.

\begin{definition}
\label{def:spider}
Let $G$ be a semicomplete digraph and
let $d \geq 0$, $l > 0$, and $w > 0$ be integers.
A {\em $(d, l, w)$-spider} is a triple $(T,
L, R)$, where $T$ is a vertex set with $|T| \geq l$,
$L$ is a family $\{L_v \mid v \in T\}$ of vertex sets,
and $R$ is a family $\{R_v \mid v \in T\}$ of vertex sets,
such that 
the following holds for each $v \in T$:
\begin{enumerate}
     \item $L_v \subseteq \Nin(v)$, 
     \item $|L_v| \geq 3l$,
     \item $\dout(u) \leq d$ for each $u \in L_v$, 
     \item $R_v \subseteq \Nout(v)$, 
     \item $|R_v| \geq 3l$, and  
     \item $\dout(u) \geq d + w$ for each $u \in R_v$.
\end{enumerate}
We will sometimes refer to a $(d, l, w)$-spider as
an $(l, w)$-spider, without specifying $d$.
\end{definition}

\begin{lemma}
\label{lem:spider_lb}
If a semicomplete digraph $G$ has an
$(l, w)$-spider then $\pw(G) > \min\{l, w\}$. 
\end{lemma}
\begin{proof}
Let $(T, L, R)$ be a $(d, l, w)$-spider of $G$.
Let $T'$ be an arbitrary subset of $T$ with
$|T'| = l$.
For each $v \in T'$, select $l_v \in L_v$ and
$r_v \in R_v$ so that,
for each distinct pair $u, v \in T'$,
we have 
$\{u, l_u, r_u\} \cap \{v, l_v, r_v\} = \emptyset$.
Since $|L_v| \geq 3l$ and 
$|R_v| \geq 3l$ for each $v \in T$, 
such a selection can trivially be done in 
a greedy manner. 
We have a set of $l$ pairwise vertex-disjoint
paths from $\Voutle{d}(G)$ to
$\Voutge{d + w + 1}(G)$ and hence
by Lemma~\ref{lem:disjoint},
we have $\pw(G) > \min\{l, w\}$.
\qed
\end{proof}

The following lemma shows that spiders capture what we intend
them to capture. 
\begin{lemma}
\label{lem:degree-spider}
Suppose $G$ has a $(2l,w)$-degree tangle $T$.
Then, $G$ has either a tame $(l, w)$-degree tangle or
an $(l, w)$-spider.
\end{lemma}
\begin{proof}
Let $U = \{v \in T \mid \wild(v) \leq 3l + w + 2\pw(G)\}$.
If $|U| \geq l$ then $U$ contains a tame $(l, w)$-degree
tangle and we are done. So, suppose otherwise.
Let $d$ be such that $T \subseteq \Voutge{d} \cap \Voutle{d + w}$. 
For each $v \in T \setminus U$,
let $L_v = \Voutle{d} \cap \Nin(v)$ and 
$R_v = \Voutge{d + w} \cap \Nout(v)$.
Fix $v \in T \setminus U$. 
As $\wild(v) > 3l  + w + 2\pw(G)$, we have, 
by Lemma~\ref{lem:wild},
\begin{eqnarray*}
|L_v| &\geq &
\wild(v) - (\dout(v) - d)  - 2\pw(G) - 1\\
& \geq & \wild(v) - w  - 2\pw(G) - 1\\
& \geq & 3l
\end{eqnarray*}
and similarly $|R_v| \geq 3l$. 
Therefore, the triple $(T \setminus U, L, R)$ is a $(d, l, w)$-spider.
\qed
\end{proof}

We similarly define the tameness of matching tangles. 
\begin{definition}
\label{def:matching-tame}
We say that a $(d, l, w)$-matching tangle $(T_1, T_2)$ of
$G$ is tame if  
\begin{enumerate}
  \item $\wild(v) \leq 3l + d + w - \dout(v) + 2\pw(G)$
  for each $v \in T_1$ and
  \item $\wild(v) \leq 3l + \dout(v) - d + 2\pw(G)$
  for each $v \in T_2$. 
\end{enumerate}
\end{definition}

\begin{lemma}
\label{lem:matching-spider}
Suppose $G$ has a $(d, 3l, w)$-matching tangle $(T_1, T_2)$. 
Then, $G$ has either a tame $(d, l, w)$-matching tangle or
a $(d, l, w)$-spider.
\end{lemma}
\begin{proof}
Let 
$I_1
 = \{v \in T_1 \mid \wild(v) > 3l + d + w - \dout(v) + 2\pw(G)\}$ and
$I_2 = \{v \in T_2 \mid \wild(v) > 3l + \dout(v) - d + 2\pw(G)\}$.
If $|I_1| \leq l$ and $|I_2| \leq l$ then
there is some $T_1' \subseteq T_1 \setminus I_1$ and
$T_2', \subseteq T_2 \setminus I_2$ with $|T_1'| = |T_2'| = l$
such that there is a matching from $T_1'$ to $T_2'$ by
edges of $G$: $(T_1', T_2')$ is a tame $(d, l, w)$-matching tangle.

Suppose otherwise.
We first consider the case where $|I_1| > l$.
For each $v \in I_1$, 
let $L_v = \Voutle{\dout(v)} \cap \Nin(v)$ and
$R_v = \Voutge{d + w} \cap \Nout(v)$. 
Applying Lemma~\ref{lem:wild} and using the assumption
$\wild(v) > 3l + d + w - \dout(v) + 2\pw(G)$ , we have
\begin{eqnarray*}
|L_v| & \geq & \wild(v) - 2\pw(G) - 1 \\
& \geq & 3l + d + w - \dout(v)\\ 
& \geq & 3l 
\end{eqnarray*}
and 
\begin{eqnarray*}
|R_v| & \geq & \wild(v) - (d + w - \dout(v)) - 2\pw(G) \\
& \geq & 3l.\\
\end{eqnarray*}
Therefore, $(I_1, L, R)$ is a $(d, l, w)$-spider.
In the case $|I_2| > l$, we have a
$(d, l, w)$-spider similarly constructed on $I_2$.
\qed
\end{proof}

We also need to define the tameness of spiders. 
\begin{definition}
\label{def:spider-tame}
Let $(T, L, R)$ be a $(d, l, w)$-spider.
We say that a vertex $u \in \bigcup_{v \in T} L_v$
is {\em tame}
(relative to the parameters $d$, $l$, and $w$) 
if $\wild(u) \leq 3l + d + w - \dout(u) + 2\pw(G)$.
Similarly, $u \in \bigcup_{v \in T} R_v$
is {\em tame} if 
$\wild(u) \leq 3l + \dout(u) - d + 2\pw(G)$.
We let $L_v^{\rm tame}$ and 
$R_v^{\rm tame}$ denote the set of
tame vertices in $L_v$ and $R_v$ respectively.
We say that a $(d, l, w)$-spider $(T, L, R)$
is {\em tame} if $|L_v^{\rm tame}| \geq 2l$
and $|R_v^{\rm tame}| \geq 2l$ for
every $v \in T$.

We say that a $(l, w)$-spider is tame, if it is a tame
$(d, l, w)$-spider for some $d$.
\end{definition}

\begin{lemma}
\label{lem:spider}
Let $G$ be a semicomplete digraph
and suppose that $G$ has an $(l, w)$-spider,
where $w > 0$.
Then, $G$ has a tame $(l, w')$-spider for
some $w' \geq w$.
\end{lemma}
\begin{proof}
Suppose $G$ has a $(d, l, w)$-spider $(T, L, R)$.
We may assume that $w$ is the largest possible given $l$:
for every $(d', l, w')$-spider of $G$, we have
$w' \leq w$. Under this assumption, we show that
the spider $(T, L, R)$ is tame.
For contradiction, suppose not. 
We consider the case where there is
some $v \in T$ such that $|L_v^{\rm tame}| < 2l$;
the case where there is
some $v \in T$ such that $|R_v^{\rm tame}| < 2l$
is similar.
Fix such $v$ and 
let $U = L_v \setminus L_v^{\rm tame}$.
Since $|L_v| \geq 3l$ by the definition
of a spider, we have $|U| \geq l$.
Let $u$ be an arbitrary member of $U$ and let 
$w_u = d + w - \dout(u)$.
Since $\dout(u) \leq d$ by the definition of a spider,
we have $w_u \geq w$.

Since $u$ is not tame, we have 
$\wild(u) > 3l + d + w - \dout(u) + 2\pw(G)
= 3l + w_u + 2\pw(G)$.
Let $L'_u = \Voutle{\dout(u) - w_u} \cap \Nin(u)$ and
$R'_u = \Voutle{\dout(u) + w_u} \cap \Nout(u)$.
We apply Lemma~\ref{lem:wild} and have
\begin{eqnarray*}
|L'_u| & \geq & \wild(u) - w_u - 2\pw(G) - 1\\
& \geq & 3l
\end{eqnarray*}
and 
\begin{eqnarray*}
|R'_u| & \geq & \wild(u) - w_u - 2\pw(G) \\
& \geq & 3l.
\end{eqnarray*}
Since $\dout(u) - w_u = d + w - 2w_u \leq d - w$ and 
$\dout(u) + w_u = d + w$
holds for every $u \in U$, 
we have $L'_u \subseteq \Voutle{d - w}$ and
$R'_u \subseteq \Voutge{d + w}$ for every $u \in U$.
Therefore, $(U, L', R')$ is a $(d - w, l, 2w)$-spider, 
contradicting the choice of $w$.
\qed
\end{proof}

To continue our scenario, we invoke the following result due
to Pilipczuk.

\begin{lemma}(\cite{Pilipczuk12}, Theorem~32)
\label{lem:tangle-Pil}
There exists an algorithm, which given a semicomplete digraph $G$ and integers
$k$ and $l \geq 5k$, in time $O(k|V(G
)|^2)$ outputs one of the following:
\begin{itemize}
  \item an $(l + 2, k)$-degree tangle in $G$;
  \item a $(k + 1, k)$-matching tangle in $G$;
  \item a path decomposition of $G$ of width at most $(l + 2k)$.
\end{itemize}
\end{lemma}

The following lemma, building on this lemma and previous lemmas,
shows that a semicomplete digraph of large pathwidth has
a tame tangle or a spider.
\begin{lemma}
\label{lem:get-tangle}
Let $K$ be a positive integer and $G$ a semicomplete
digraph with $\pw(G) \geq 128K$. 
Then, $G$ has at least one of the following:
\begin{enumerate}
  \item a tame $(46K, 18K)$-degree tangle;
  \item a $(6K, 18K)$-spider;
  \item a tame $(6K, 18K)$-matching tangle.
\end{enumerate}
\end{lemma}
\begin{proof}
We apply Lemma~\ref{lem:tangle-Pil} to $G$ with
$l = 92K - 2$ and $k = 18K$. 
Since $G$ 
does not have a path-decomposition of width 
$l + 2k = 92K - 2 + 36K = 128K - 2$,
the algorithm finds either a
$(92K, 18K)$-degree tangle of $G$, or an
$(18K + 1, 18K)$-matching tangle of $G$.
In the first case,
by Lemma~\ref{lem:degree-spider}, $G$ has either a
tame $(46K, 18K)$-degree tangle or a $(46K, 18K)$-spider,
which certainly contains a $(6K, 18K)$-spider.
In the second case, $G$ has a $(18K, 18K)$-matching tangle 
and, hence by Lemma~\ref{lem:matching-spider},
either a tame $(6K, 18K)$-matching tangle or a
$(6K, 18K)$-spider.
\qed 
\end{proof}

\begin{lemma}
\label{lem:degree-sample}
Let $h$ be a positive integer. Then, there is
some positive integer $k_h$ such that the following holds.
Let $k \geq k_h$ be an integer and let $K = (h + 1)k$.
Let $G$ be an $h$-semicomplete digraph and
suppose a semicomplete supergraph $G'$ of $G$ with vertex set
$V(G)$ and with $\pw(G') \leq 140K$ 
has a tame $(46K, 18K)$-degree tangle.
Then $G$ has a semicomplete subgraph with
a $(21k, 10k)$-degree tangle.
\end{lemma}
\begin{proof}
Let $T$ be a tame $(46K, 18K)$-degree tangle of $G'$. 
Let $\hat{G}$ denote the
complement of the undirected graph underlying $G$. The maximum degree of
$\hat{G}$ is $h$ or smaller. We apply Theorem~\ref{thm:sample-indep}
to $\hat{G}$ to obtain a random independent set $I$ of
$\hat{G}$.  The probability of each vertex being in $I$ is 
$p = \frac{1}{2(h + 1)}$. For each $S \subseteq V(G)$,
the expectation of $|S \cap I|$ is $p|S|$ and
the probability of deviations is bounded as in 
Theorem~\ref{thm:sample-indep}.

That $I$ is independent in $\hat{G}$ implies
that $G[I]$, which equals $G'[I]$, is semicomplete. 
We show that $T \cap I$ contains a
$(21k, 10k)$-tangle of $H = G[I]$ with high probability.

We call the event $|T \cap I| < 21k$ {\em the bad event
on $|T \cap I|$}.

Since $\ex[|T \cap I|] = 46pK = 23k$, 
the probability of this bad event is at most 
\begin{eqnarray*}
\prob\left(23k - |T \cap I| > 2k\right) 
& \leq &\exp\left(-\frac{4k^2}{9\cdot 46K}\right)\\
& = & \exp\left(-\frac{2k}{207(h + 1)}\right)
\end{eqnarray*}
by Theorem~\ref{thm:sample-indep}.

Let $d$ be such that $T \subseteq \Voutge{d}(G') \cap \Voutle{d + 18K}(G')$.
For each $v \in T \cap I$, we evaluate $d^+_H(v)$ as follows.
\begin{eqnarray*}
d^+_H(v) & = & |\Nout_{G'}(v) \cap I|\\
& = & |\Voutle{d}(G') \cap I| -
|(\Voutle{d}(G') \setminus \Nout_{G'}(v)) \cap I| +
|(\Voutge{d + 1}(G') \cap \Nout_{G'}(v)) \cap I|.
\end{eqnarray*} 

The deviation of the first term is common
for all $v$:
$\varDelta = |\Voutle{d}(G') \cap I| - 
\ex[|\Voutle{d}(G') \cap I|] = 
|\Voutle{d}(G') \cap I| - p|\Voutle{d}(G')|$.
Therefore, we are concerned with the deviations of other
terms depending on $v$.

Let $X_v = \Voutle{d}(G') \setminus \Nout_{G'}(v)$ and
$Y_v = \Voutge{d + 1}(G') \cap \Nout_{G'}(v)$
for each $v \in T$. As the $(46K, 18K)$-degree tangle 
$T$ of $G'$ is tame,
we have $\wild(v) \leq 3\cdot 46K + 18K + 2\pw(G') 
\leq 436K$ and hence 
\begin{eqnarray*}
|X_v| & \leq & |\Voutle{d^+_{G'}(v)}(G') \setminus \Nout_{G'}(v)| \\
& = & \wild(v) \\
& \leq & 436K.
\end{eqnarray*}
Since
\begin{eqnarray}
|\Voutle{d}(G')| & = & n - |\Voutge{d + 1}(G')| \nonumber\\
& \geq & n - (n - d - 1 + 2\pw(G')) \nonumber\\
& \geq & d + 1 - 280K \label{eqn:XvYv1}
\end{eqnarray}
by Lemma~\ref{lem:degree-interval} and hence 
\begin{eqnarray}
|\Voutle{d}(G')| & \geq & d^+_{G'}(v) - 298K + 1, \label{eqn:XvYv}
\end{eqnarray}
we have 
\begin{eqnarray*}
|Y_v| & =  & d^+_{G'}(v) - (|\Voutle{d}(G')| - |X_v|) \\
& \leq & 298K + 436K\\
& \leq & 734K.
\end{eqnarray*}

Call the event 
$||X_v \cap I| - p|X_v|| > \frac{k}{4}$
{\em the bad event on $X_v$} and the 
event $||Y_v \cap I| - p|Y_v|| > \frac{k}{4}$
{\em the bad event on $Y_v$}.
By Theorem~\ref{thm:sample-indep}, 
the probability of the bad event on $X_v$ is smaller than
\begin{eqnarray*}
2\exp\left(-\frac{k^2}{4^2 \cdot 9|X_v|}\right)
& \leq & 2\exp\left(-\frac{k}{62784(h + 1)}\right)
\end{eqnarray*}
and, similarly, 
the probability of the bad event on $Y_v$ is smaller than 
\begin{eqnarray*}
2\exp\left(-\frac{k}{105696(h + 1)}\right).
\end{eqnarray*}
Therefore, setting say, $k_h = 10^7(h + 1)^2$,
it follows from our assumption $k \geq k_h$ that,
with probability close to 1, none of the
bad events listed above occurs.

Assume none of those bad events occur.
Recall that $\varDelta = |\Voutle{d}(G') \cap I| -
p|\Voutle{d}(G')|$.
Then, for each $v \in T \cap I$, we have
\begin{eqnarray*}
d^+_{H}(v) & = & |\Voutle{d}(G') \cap I| - |X_v \cap I|
+ |Y_v \cap I| \\
& \leq & p|\Voutle{d}(G')| + \varDelta - 
p|X_v| + \frac{k}{4} + 
p|Y_v| + \frac{k}{4} \\
& \leq & p d^+_{G'}(v) + \varDelta + \frac{k}{2}
\end{eqnarray*}
and, similarly, 
\begin{eqnarray*}
d^+_H(v) & \geq & 
p d^+_{G'}(v) + \varDelta - \frac{k}{2}.
\end{eqnarray*}
Therefore, for each $v \in T \cap I$, we have
\begin{eqnarray*}
pd + \varDelta - \frac{k}{2} 
\leq  d^+_H(v)
& \leq & p(d + 18K) + \varDelta + \frac{k}{2} \\ 
& = & pd + \varDelta + \frac{19k}{2}.
\end{eqnarray*}
Therefore, $T \cap I$ contains a $(21k, 10k)$-degree
tangle of $H$.
\qed
\end{proof}

\begin{lemma}
\label{lem:spider-sample}
Let $h$ be a positive integer. Then, there is
some positive integer $k_h$ such that the following holds.
Let $k \geq k_h$ be an integer and let $K = (h + 1)k$.
Let $G$ be an $h$-semicomplete digraph and
suppose a semicomplete supergraph $G'$ of $G$ with vertex set
$V(G)$ and with $\pw(G') \leq 140K$
has a $(6K, 18K)$-spider.
Then $G$ has a semicomplete subgraph with
a $(k, k)$-spider.
\end{lemma}
\begin{proof}
Since $G'$ has a $(6K, 18K)$-spider, by Lemma~\ref{lem:spider}, 
it has a tame $(6K, w)$-spider for some $w \geq 18K$.
The approach is similar to the proof of Lemma~\ref{lem:degree-sample}.
The only essential difference is that
the wildness of a vertex in the spider
may not be $O(K)$ and the deviation of
its out-degree in the sampled subgraph may
be large. This is not
an essential problem, however, since such a vertex
with large wildness has, by the definition of
tame spiders, the original out-degree 
far away from the range to be avoided
and therefore a large deviation is affordable.

Let $(T, L, R)$ be a tame $(d, 6K, w)$-spider of $G'$,
where $w \geq 18K$.
As in the proof of Lemma~\ref{lem:degree-sample},
let $\hat{G}$ be the undirected graph underlying $G$, 
$p = \frac{1}{2(h + 1)}$, $I$ the set of independent
vertices of $\hat{G}$ sampled with probability $p$
applying Theorem~\ref{thm:sample-indep}, and $H = G'[I] = G[I]$.
Let $T' = T \cap I$ and, 
for each $v \in T'$, let
$L'_v = L_v^{\rm tame} \cap I$ and 
$R'_v = R_v^{\rm tame} \cap I$.
Our goal is to show that 
$(T', L', R')$ is a $(k, k)$-spider of
$H = G'[I]$ with 
high probability.  

For this to happen, we need to have
$|T'| \geq k$ and, for some $d'$
and for each $v \in T'$, 
\begin{enumerate}
     \item $L_v' \subseteq N^-_H(v)$, 
     \item $|L_v'| \geq 3k$,
     \item $d^+_H(u) < d'$ for each $u \in L_v'$, 
     \item $R_v' \subseteq N^+_H(v)$,  
     \item $|R_v'| \geq 3k$, and  
     \item $d^+_H(u) > d' + k$ for each $u \in R_v'$.
\end{enumerate}

We list
``bad'' events below that could prevent
the above conditions from being satisfied.
We show that the probability of each of those events is
$\exp(-\Omega(\frac{k}{h}))$ and, since
the number of those events is obviously $O(kh)$,
the probability is close to 1 that none of these
events occurs under the assumption $k \geq k_h$
if $k_h$ is large enough.
We also confirm that if   
none of those events occurs then 
the above conditions for $(T', L', R')$
being a $(k, k)$-spider are all satisfied.

Since most of the analysis below is
similar to the one we did for Lemma~\ref{lem:degree-sample},
we omit some details, using $\Omega$ notation rather
than giving explicit constants in probability bounds, 
and emphasize what is different.

First consider 
the event that $|T \cap I| < k$.
Since $|T| \geq 6K$ and hence $\ex[|T \cap I|] \geq  
3k$, the probability of this event
is $\exp(-\Omega(\frac{k}{h}))$.
Next consider, for each $v \in T$, the event that
$|L_v \cap I| < 3k$ or $|R_v \cap I| < 3k$.
Since $|L_v| \geq 9K$ and $|R_v| \geq 9K$,
the probability of this event is also
$\exp(-\Omega(\frac{k}{h}))$.
If none of these events occurs, all conditions
enumerated above are satisfied but those on the out-degrees on
vertices in $\bigcup_{u \in T'} L_u'$
and in $\bigcup_{u \in T'} R_u'$.

We proceed to events that may cause intolerable
deviations of the out-degrees of those vertices.

For each $v \in \bigcup_{u \in T} 
(L_u^{\rm tame} \cup R_u^{\rm tame})$,
let $X_v = \Voutle{d}(G') \setminus N^+_{G'}(v)$ and
$Y_v = \Voutge{d + 1}(G') \cap N^+_{G'}(v)$.
As in the proof Lemma~\ref{lem:degree-sample}, we evaluate $d^+_H(v)$ 
(assuming $v \in I$), as follows:
\begin{eqnarray*}
d^+_H(v) & = & |N^+_{G'}(v) \cap I|\\
& = & |\Voutle{d}(G') \cap I| - |X_v \cap I|
+ |Y_v \cap I|.
\end{eqnarray*} 
The deviation of the first term is common
for all $v$:
$\varDelta = |\Voutle{d}(G') \cap I| - 
\ex[|\Voutle{d}(G') \cap I|] = 
|\Voutle{d}(G') \cap I| - p|\Voutle{d}(G')|$.

Therefore, our bad events concern about
the deviations of $|X_v \cap I|$ and 
of $|Y_v \cap I|$ from their expectations.

First consider $v \in \bigcup_{u \in T}L_u^{\rm tame}$. 
From the tameness condition and by 
Lemma~\ref{lem:degree-interval},
we have 
\begin{eqnarray*}
|X_v| & = & |\Voutle{d}(G') \setminus N^+_{G'}(v)| \\
& \leq & |\Voutge{d^+_{G'}(v) + 1}(G') \cap
\Voutle{d}(G')| + 
|\Voutle{d^+_{G'}(v)}(G') \setminus N^+_{G'}(v)| \\
& \leq & d - d^+_{G'}(v) + 2\pw(G') + \wild(v) \\
& \leq & 3(6K) + 2(d - d^+_{G'}(v)) + w + 4\pw(G')\\ 
& \leq & 578K + w + 2(d - d^+_{G'}(v)), 
\end{eqnarray*}
and using (\ref{eqn:XvYv}), 
\begin{eqnarray*}
|Y_v| & =  & d^+_{G'}(v) - (|\Voutle{d}(G')| - |X_v|) \\
& \leq & 298K + |X_v|\\
& \leq & 876K + w + 2(d - d^+_{G'}(v)).
\end{eqnarray*}

Note that neither $w$ nor $d - d^+_{G'}(v)$ is necessarily $O(K)$.
Our bad events on $X_v$ and $Y_v$ here are 
that $|X_v \cap I| < p|X_v| - \max\{\frac{pw}{6}, 
\frac{p}{2}(d - d^+_{G'}(v))\}$ and that
$|Y_v \cap I| > p|Y_v| + \max\{\frac{pw}{6}, 
\frac{p}{2}(d - d^+_{G'}(v))\}$ respectively.

If $\frac{w}{6} \geq \frac{1}{2}(d - d^+_{G'})$, then,
noting that $w \geq 18K$ and hence $|X_v| = O(w)$ and 
$|Y_v| = O(w)$, the probability 
of each of these events is 
\begin{eqnarray*}
\exp\left(-\Omega\left(\frac{p^2 w^2}{|X_v|}\right)\right)
& = &  \exp\left(-\Omega\left(\frac{p^2 w^2}{w}\right)\right)\\
& = & \exp\left(-\Omega\left(p^2 K\right)\right)\\
& = & \exp\left(-\Omega\left(\frac{k}{h}\right)\right).
\end{eqnarray*}
The other case is similar and the probability of
each of these events is $\exp(-\Omega(\frac{k}{h}))$ in either case.
We conclude that, with probability close to 1, none of the
above bad events occurs.

We analyze the out-degree of each vertex 
$v \in \bigcup_{u \in T'} L_u'$
assuming that the bad event on neither $X_v$ nor
$Y_v$ occurs. We have
\begin{eqnarray*}
d^+_H(v) & = & |\Voutle{d}(G') \cap I| - |X_v \cap I|
+ |Y_v \cap I|\\
& = & p(|\Voutle{d}(G')| -|X_v| + |Y_v|) + \varDelta
+ (p|X_v| - |X_v \cap I|)
+ (|Y_v \cap I| - p|Y_v|)\\
& = & p d^+_{G'}(v) + \varDelta +
(p|X_v| - |X_v \cap I|) + 
(|Y_v \cap I| - p|Y_v|).
\end{eqnarray*} 
The sum of the last two terms is
at most $2\max\{\frac{pw}{6}, \frac{p}{2}(d - d^+_{G'}(v))\}
=  \max\{\frac{pw}{3}, p(d - d^+_{G'}(v)) \leq 
\frac{pw}{3} + p(d - d^+_{G'}(v))$.
Therefore, we have 
\begin{eqnarray}
\label{eqn:upper_left}
d^+_H(v) & \leq & pd + \varDelta + \frac{pw}{3}
\end{eqnarray}
for each $v \in \bigcup_{u \in T'} L_u'$.

Next consider a vertex $v \in \bigcup_{u \in T}R_u^{\rm tame}$.
From the tameness condition, 
we have 
\begin{eqnarray*}
|X_v| & = & |\Voutle{d}(G') \setminus N^+_{G'}(v)| \\
& \leq & \wild(v) \\
& \leq & 3\cdot 6K + d^+_{G'}(v) - d  +
2\pw(G')\\
& \leq & 298K + d^+_{G'}(v) - d, 
\end{eqnarray*}
and using (\ref{eqn:XvYv1}), 
\begin{eqnarray*}
|Y_v| & =  & d^+_{G'}(v) - (|\Voutle{d}(G')| - |X_v|) \\
& \leq & 280K + d^+_{G'}(v) - d + |X_v|\\
& \leq & 578K + 2(d^+_{G'}(v) - d).
\end{eqnarray*}

Our bad events on $X_v$ and $Y_v$ here are 
that $|X_v \cap I| > p|X_v| + \frac{p}{6}(d^+_{G'}(v) - d)$ and that
$|Y_v \cap I| < p|Y_v| - \frac{p}{6}(d^+_{G'}(v) - d)$ respectively.
Since $d^+_{G'}(v) - d \geq w \geq 18K$, 
the probability of each of these bad events is $\exp(-\Omega(\frac{k}{h}))$
and therefore, with probability close to 1, none of these bad events occurs
for any $v \in \bigcup_{u \in T}R_u^{\rm tame}$.

We analyze the out-degree of each vertex 
$v \in \bigcup_{u \in T'} R_u'$ assuming that none of the bad events occurs.
We have 
\begin{eqnarray*}
d^+_H(v) & = & 
pd^+_{G'}(v) + \varDelta
- (|X_v \cap I| - p|X_v|) - (p|Y_v| - |Y_v \cap I|)
\end{eqnarray*} 
as before and the sum of the last two terms, neglecting signs,
is at most 
$\frac{2p}{6}(d^+_{G'}(v) - d) =  \frac{p}{3}(d^+_{G'}(v) - d)$.
Therefore, we have 
\begin{eqnarray}
d^+_H(v) & \geq & pd^+_{G'}(v) + \varDelta - \frac{p}{3}(d^+_{G'}(v) -
d)\nonumber\\
& \geq & pd + \varDelta  + \frac{2p}{3}(d^+_{G'}(v) - d)\nonumber\\
& \geq & pd + \varDelta + \frac{2p}{3}w
\label{eqn:lower_right}
\end{eqnarray}
for each $v \in \bigcup_{u \in T'} R_u'$.
From (\ref{eqn:upper_left}) and (\ref{eqn:lower_right}),
we have 
\begin{eqnarray*}
\min\{d^+_H(v) \mid v \in \bigcup_{u \in T'} R_u'\} - 
\max\{d^+_H(v) \mid v \in \bigcup_{u \in T'} L_u'\} 
& \geq & \frac{pw}{3}\\
& \geq & 6pK\\
& = & 3k.
\end{eqnarray*}
Therefore, $(T', L', R')$ is a 
$(d', k, k)$-spider of $H$ for some $d'$.
\qed
\end{proof}

\begin{lemma}
\label{lem:matching-sample}
Let $h$ be a positive integer. Then, there is
some positive integer $k_h$ such that the following holds.
Let $k \geq k_h$ be an integer and let $K = (h + 1)k$.
Let $G$ be an $h$-semicomplete digraph and
suppose a semicomplete supergraph $G'$ of $G$ with vertex set
$V(G)$ and $\pw(G') \leq 140K$ has a tame 
$(6K, 18K)$-matching tangle $(T_1, T_2)$.  Suppose moreover
that the matching bijection $\phi$ of this tangle
is such that the edge $(v, \phi(v))$ of $G'$ for each $v \in T_1$
is in fact an edge of $G$.

Then $G$ has a semicomplete subgraph that has
a $(k, k)$-matching tangle.
\end{lemma}
\begin{proof}

Let $\hat{G}$ be
the complement of the undirected graph underlying
$G$.  Let $\tilde{G}$ be obtained from $\hat{G}$ by 
contracting the doubleton $\{v, \phi(v)\}$ into
a vertex, say $t_v$, for each $v \in T_1$. 
Let $T = \{t_v \mid v \in T_1\}$.
Note that the maximum degree of $\tilde{G}$ is
$2h$ or smaller. Similarly to Lemma~\ref{lem:spider-sample}, 
we use Theorem~\ref{thm:sample-indep} to obtain
an independent set $I'$ of $\tilde{G}$ where
the probability of each $v \in V(\tilde{G})$
belonging to $I'$ is $\frac{1}{2(2h + 1)}$. 
Let $H = G[I]$ where
$I = 
((V(\hat{G}') \setminus T) \cap I')
\cup \{v , \phi(v) \mid v \in T_1, t_v \in I'\}$.
As $I'$ is independent in $\hat{G}'$, 
$I$ is independent in $\hat{G}$ and hence 
$G'[I] = G[I]$ is semicomplete.
By an analysis similar to the one in Lemma~\ref{lem:spider-sample}, 
setting $k_h$ large enough, 
we have 
$|T \cap I| \geq k$ and
$\min_{v \in T_2 \cap I} d^+_H(v) - 
\max_{v \in T_1 \cap I} d^+_H(v) \geq k$ with
probability close to 1.
When this happens, $(T_1 \cap I, T_2 \cap I)$ contains
a $(k, k)$-matching tangle of $H$.
\qed
\end{proof}

We are now ready to prove Theorem~\ref{thm:comb}.
Fix positive integer $h$. Let $k_h$ be a constant large
enough as required in Lemmas~\ref{lem:degree-sample}, \ref{lem:spider-sample},
and \ref{lem:matching-sample}.
We set $f(k, h) = 128(h + 1)k$ for $k \geq k_h$ and
$f(k, h) = f(k_h, h)$ for $k < k_h$.

Let $G$ be an $h$-semicomplete
digraph of pathwidth at least $f(k, h)$.
In the following proof that $G$ contains
a semicomplete subgraph of pathwidth at least $k$,
we assume $k \geq k_h$; otherwise we would
prove that $G$ contains a semicomplete subgraph
of pathwidth at least $k_h \geq k$.
We set $K = (h + 1)k$ for readability.

List the vertices of $G$ as $v_1$, \ldots, $v_n$,
in the non-decreasing order of out-degrees. Let $G'$ be
the semicomplete digraph obtained from $G$ by adding
edge $(v_i, v_j)$ for each pair $i > j$ such that
neither $(v_i, v_j)$ nor $(v_j, v_i)$ is an edge of $G$.
By our assumption, $\pw(G') \geq \pw(G)$ is at least $128K$.
We assume below that $\pw(G') \leq 140K$;
if this assumption does not hold, we choose $k' \geq k$ such that
$128(h + 1)k' \leq \pw(G') \leq 140(h + 1)k'$ and prove that
$G$ has a semicomplete subgraph of pathwidth $\geq k'$.

Applying Lemma~\ref{lem:get-tangle}, we obtain 
a tame $(46K, 18K)$-degree tangle, a tame
$(6K, w)$-spider for some $w \geq 18K$, or a tame $(6K, 18K)$-matching
tangle of $G'$.  

If $G'$ has a tame $(46K, 18K)$-degree tangle, then 
$G$ has a semicomplete subgraph that 
contains $(21k, 10k)$-degree tangle, by
Lemma~\ref{lem:degree-sample}.
If $G'$ has a tame $(6K, w)$-spider for $w \geq 18K$, then 
$G$ has a semicomplete subgraph that 
contains a $(k, k)$-spider, by Lemma~\ref{lem:spider-sample}.
Finally, suppose $G'$ has a $(6K, 18K)$-matching tangle $(T_1, T_2)$
with matching bijection $\phi$. We observe that, for each $v \in T_1$,
the edge $(v, \phi(v))$ of $G'$ is in fact an edge
of $G$, since 
\begin{eqnarray*}
\doutG(\phi(v)) &\geq & d^+_{G'}(\phi(v)) - h
 \geq  d^+_{G'}(v) + 18K - h
 >  d^+_{G'}(v)\\
&\geq & \doutG(v)
\end{eqnarray*}
and the edge addition rule for constructing $G'$ from $G$ dictates
that if an edge between $v$ and $\phi(v)$ is added 
then it must be from 
$\phi(v)$ to $v$. Therefore, 
Lemma~\ref{lem:matching-sample} applies and $G$ has a  
semicomplete subgraph with a $(k, k)$-matching tangle.

In either case, we conclude that $G$ contains a semicomplete
subgraph of pathwidth at least $k$. This completes the
proof of Theorem~\ref{thm:comb}.

\section{Proof of Theorem~\ref{thm:sample-indep}}
\label{sec:sample-indep}
The goal of this section is to prove Theorem~\ref{thm:sample-indep}, 
which we restate below.
Graphs are undirected in this section and we use
the following notation.
For each $v \in V(G)$,
$N_G(v)$ is the set of neighbors of $v$ and 
$N_G[v] = N_G(v) \cup \{v\}$;
for each $U \subseteq V(G)$, 
$N_G[U] = \bigcup_{u \in U} N_G[u]$ 
and $N_G(U) = N_G[U] \setminus U$.

\begingroup
\def\thetheorem{\ref{thm:sample-indep}}
\begin{theorem}
Let $G$ be an undirected graph on $n$ vertices 
with maximum degree $d$ or smaller.
Let $p = \frac{1}{2d + 1}$.
Then, it is possible to sample a set $I$ of independent vertices of $G$ so that
$\prob(v \in I) = p$ for each $v \in V(G)$ and,   
for each $S \subseteq V(G)$,  
we have  
\begin{eqnarray*}
 \prob(|S \cap I| > p|S|+ t) <
 \exp\left(-\frac{t^2}{9|S|}\right)
 \end{eqnarray*}
 and 
 \begin{eqnarray*}
 \prob(|S \cap I| <p|S| - t) <
 \exp\left(-\frac{t^2}{9|S|}\right).
 \end{eqnarray*}
\end{theorem}
\addtocounter{theorem}{-1}
\endgroup

\smallskip
A naive sampling method is to keep a set $V$ of candidate vertices and
repeatedly pick a random vertex from $V$ to add to $I$,   
removing the selected vertex and all of its neighbors from $V$.
This procedure would 
produce an independent set of cardinality at least $n / (d + 1)$. 
The exact probability of each vertex being in $I$, however, would depend
on the structure of $G$. To achieve the uniform probability as claimed
in the above theorem, we sample, at each step, from a $d$-regular
supergraph of $G[V]$ rather than from $G[V]$ itself.

We need the following theorem on regular completion of graphs due to 
Erd\H{o}s and Kelly.

\begin{theorem}
\label{thm:reg_comp}
\cite{EK63}
Let $G$ be an undirected graph on $n$ vertices
and $d$ an integer such that $d_G(v) \leq d$
for every $v \in V(G)$. Let $t = \sum_{v \in V(G)} (d - d_G(v))$.
Then, there is a $d$-regular graph on
$n + m$ vertices that has $G$ as an induced subgraph
if and only if $m$ satisfies all of the following four conditions:
\newline\noindent (1) $md \geq t$;
\newline\noindent (2) $m^2 - m(d + 1) + t \geq 0$;
\newline\noindent (3) $m \geq d - d_G(v)$ for every $v \in V(G)$; and
\newline\noindent (4) $(n + m)d$ is an even integer.
\end{theorem}

Akiyama {\it et al.} \cite{AEH83} proved that, 
for every graph $G$ on $n$ vertices with maximal degree $d$ or smaller,
there is a $d$-regular graph on $N \leq n + d + 2$ vertices ($N \leq n + d + 1$
if $nd$ is even) that contains $G$ as a (not necessarily induced) subgraph.
The following lemma states that every integer $N \geq n + d + 1$ with 
$Nd$ even has that property.
The proof is, naturally, analogous to the one in \cite{AEH83}.
\begin{lemma}
\label{lem:paving}
Let $G$ be a graph on $n$ vertices with maximum degree $d$ or smaller
and $N$ an arbitrary integer such that $N \geq n + d + 1$ and $Nd$ is even.  
Then, there is a $d$-regular graph on $N$ vertices that contains $G$ as a 
subgraph.
\end{lemma}
\begin{proof}
Let $H$ be a maximal graph on $V(G)$ with
maximum degree $d$ that contains all the edges of $G$.
Let $D = \{v \in V(G) \mid d_{H}(v) < d\}$.
From the maximality of $H$, $D$ must be a clique of $H$ and
hence $|D| \leq d$. It trivially follows that 
$t = \sum_{v \in V(G)} (d - d_{H}(v)) \leq d^2$.
Setting $m = N - n \geq d + 1$, conditions (1), (2) and
(3) of Theorem~\ref{thm:reg_comp} are trivially satisfied. Condition
(4) is also satisfied as we are assuming $Nd$ is even.
Thus, we may apply Theorem~\ref{thm:reg_comp} to $H$ to have 
a $d$-regular graph that contains $H$ and hence $G$ as
a subgraph.
\qed 
\end{proof}

We now describe the sampling procedure of Theorem~\ref{thm:sample-indep}.
Fix a graph $G$ on
$n$ vertices with maximum degree $d$ or smaller.
Let $s = \lceil n / (d + 1) \rceil$.
We construct a sequence of pairs $(I_i, V_i)$ for
$0 \leq i \leq s$, where $\emptyset = I_0 \subseteq I_1 \subseteq \ldots
\subseteq I_s$ and $V(G) = V_0 \supseteq V_1 \supseteq \ldots
\supseteq V_s$. Our independent set $I$ is $I_s$.

Fix $i$, $0 \leq i < s$ and suppose we have constructed
$I_i$ and $V_i$. We construct 
$I_{i + 1}$ and $V_{i + 1}$ as follows.
Let $n_i = (2s - i)(d + 1)$. 
Since $i < s$, we have $n_i \geq n + d + 1 \geq |V_i| + d + 1$. Moreover, 
$n_i d$ is even as $d + 1$ divides $n_i$.
Therefore, Lemma~\ref{lem:paving} applies and there is a $d$-regular
supergraph $H_i$ of $G[V_i]$ on $n_i$ vertices. 
We pick a vertex $v$ of $H_i$ uniformly at random. 
If $v \in V_i$ then we set $I_{i + 1} = I_i \cup \{v\}$;
otherwise, we set $I_{i + 1} = I_i$.
In either case, we set 
$V_{i + 1} = V_i \setminus (\{v\} \cup N_{H_i}(v))$.
Since $H_i$ is a supergraph of $G[V_i]$, this ensures 
that $v$ is independent, in $G$, of all vertices in $V_{i + 1}$.
By a straightforward induction, $I_i$ is an independent set 
of $G$, $V_i \subseteq V(G) \setminus I_i$, and
there is no edge of $G$ between $I_i$ and $V_i$, 
for $0 \leq i \leq s$.

\begin{remark}
To make $I_i$ and $V_i$ well-defined random variables for
$0 \leq i \leq s$, we assume that the $d$-regular supergraph $H_i$
of $G[V_i]$ used above is uniquely determined from $V_i$ and $n_i$ by
some deterministic procedure relying on some predefined total order on
$V(G)$ for tie-breaking.
\end{remark}

\begin{lemma}
For each $v \in V(G)$ and $0 \leq i \leq s$,  
\begin{eqnarray*}
\prob(v \in I \mid v \in V_i) = \frac{s - i}{n_i}.
\end{eqnarray*}
\end{lemma}
\begin{proof}
The proof is by induction on $s - i$. 
The base case $i = s$ is trivial.
For the induction step, suppose $i < s$. 
Using the induction hypothesis, we have 
\begin{eqnarray*}
\label{eqn:1}
\prob(v \in I \mid v \in V_i) & = &
\prob(v \in I_{i + 1} \mid v \in V_i) +
\prob(v \in V_{i + 1} \mid v \in V_i) \prob(v \in I \mid v \in V_{i +
1}) \\
& = & \frac{1}{n_i} + \frac{n_i - (d + 1)}{n_i}\cdot\frac{s - i - 1}{n_{i +
1}}\\
& = & \frac{1}{n_i} + \frac{n_{i + 1}}{n_i}\cdot\frac{s - i - 1}{n_{i +
1}}\\
& = & \frac{s - i}{n_i}.
\end{eqnarray*}
\qed
\end{proof}
\begin{corollary}
For each $v \in V(G)$, we have
\begin{eqnarray*}
\prob(v \in I) = \frac{1}{2(d + 1)}.
\end{eqnarray*}
\end{corollary}

Therefore, we have, for each vertex set $S \subseteq V(G)$,
\begin{eqnarray*}
\ex[|S \cap I|] = \frac{|S|}{2(d + 1)}.
\end{eqnarray*} 
We show that the value $|S \cap I|$ is
sharply concentrated around its expectation, to establish
Theorem~\ref{thm:sample-indep}.
We assume $d \geq 1$ in the following analysis: the case
$d = 0$ is trivial.

Fix $S \subseteq V(G)$.
We first consider the case where $|S| \geq \frac{s}{2}$.
We define a random variable $Y_i$ for $0 \leq i \leq s$ by
\begin{eqnarray*}
Y_i = \ex[|S \cap I| \mid (I_0, V_0), (I_1, V_1), \ldots, 
(I_i, V_i)],
\end{eqnarray*}
where the expectation is conditioned on the partial outcome of the 
experiment up to the construction of $I_i$ and $V_i$.
We have 
\begin{eqnarray*}
Y_s & = &|S \cap I|,\\
Y_0 & = & \ex[|S \cap I|] = \frac{|S|}{2(d + 1)}, 
\end{eqnarray*}
and, for $0 \leq i < s$,
\begin{eqnarray*}
Y_i = \ex[Y_{i + 1} \mid (I_0, V_0), (I_1, V_1), \ldots, 
(I_i, V_i)],
\end{eqnarray*}
where the expectation is conditioned similarly to the above.
Therefore, the sequence $Y_0$, \ldots, $Y_s$ is a martingale.

We show that
\begin{eqnarray}
\label{eqn:bounded_diff1}
|Y_i - Y_{i-1}| \leq \frac{3}{2}
\end{eqnarray}
holds for $0 < i \leq s$. 
We have 
\begin{eqnarray*}
Y_i & = & |S \cap I_i| + 
\sum_{v \in S \cap V_i} \prob(v \in I \mid v \in V_i)\\
& = & |S \cap I_i| + \frac{|S \cap V_i|(s - i)}{n_i}.
\end{eqnarray*}
Since both $|S \cap V_i|$ and the fraction $(s - i) / n_i$ 
are monotone non-increasing in $i$ and 
$|S \cap I_i| - |S \cap I_{i-1}|
\leq 1$, we have $Y_i -
Y_{i-1} \leq 1$.
We also have 
\begin{eqnarray*}
Y_{i-1} - Y_i & \leq &\frac{|S \cap V_{i-1}|(s - (i-1))}{n_{i-1}}
- \frac{|S \cap V_i|(s - i)}{n_i} \\
& = & \frac{(|S \cap V_{i-1}|- 
|S \cap V_i|)(s - (i-1))}{n_{i - 1}}\\
&& + 
|S \cap V_i|\left(\frac{s - (i-1)}{n_{i-1}} - \frac{s -
i}{n_i}\right)\\
& \leq & \frac{(d + 1)(s - (i-1))}{n_{i-1}} + 
\frac{|S \cap V_i|}{n_i}\\
& \leq & \frac{s - (i-1)}{2s - (i-1)} + 
\frac{|S|}{n}\\
& \leq & \frac{3}{2}
\end{eqnarray*}
and hence (\ref{eqn:bounded_diff1}).

We use the following form of Azuma's inequality \cite{AS92}.
Let $X_0$, $X_1$, \ldots, $X_m$ be a martingale with
\begin{eqnarray*}
|X_{i + 1} - X_i| \leq 1
\end{eqnarray*}
for all $0 \leq i < m$. Let $\lambda > 0$ be arbitrary.
Then, 
\begin{eqnarray}
   \prob(X_m > X_0 + \lambda \sqrt{m}) < \exp(-\lambda^2 / 2)
\end{eqnarray}
and 
\begin{eqnarray}
   \prob(X_m < X_0 - \lambda \sqrt{m}) < \exp(-\lambda^2 / 2)
\end{eqnarray}

Applying this inequality for martingale 
$Y'_i = \frac{2}{3} Y_i$, $0 \leq i \leq s = m$,
with $\lambda = \frac{2t}{3\sqrt{s}}$, 
we have 
\begin{eqnarray*}
\prob(Y_s > Y_0 + t) & = & 
\prob(Y'_s > Y'_0 + \frac{2t}{3}) \\
& < &  \exp\left(-\frac{4t^2}{9\cdot 2s}\right)\\
& \leq & \exp\left(-\frac{t^2}{9|S|}\right)
\end{eqnarray*}
and, similarly,
\begin{eqnarray*}
\prob(Y_s < Y_0 - t) & < &  \exp\left(-\frac{t^2}{9|S|}\right),
\end{eqnarray*}
finishing the case where $|S| \geq \frac{s}{2}$.

We turn to the case where $|S| < \frac{s}{2}$.
We define a sequence $i_0$, $i_1$, \ldots, $i_m$ of indices,
where $m = 3|S|$, that depends on the outcome of the sampling,
inductively as follows.
\begin{enumerate}
  \item $i_0 = 0$.
  \item For $j > 0$, $i_j$ is the smallest $i \geq i_{j - 1}$ that satisfies
  either of the following conditions:
  \newline(1) $i = s$;
  \newline(2) $V_i \cap S \neq V_{i_{j - 1}} \cap S$;
  \newline(3) $i - i_{j - 1} \geq \frac{s}{2|S|}$.
\end{enumerate}
Note that if $i_j = s$ for some $j$, then we have
$i_{j'} = s$ for $j \leq j' \leq m$. We also note
that $i_m = s$, since, in determining $i_j$ for $1 \leq j \leq m$, 
the second condition may apply at most $|S|$ times 
and the third condition at most
$2|S|$ times, but at most $2|S| - 1$ times if 
the second condition applies at all.

We define a random variable $Z_j$ for $0 \leq j \leq m$ by
\begin{eqnarray*}
Z_j = \ex[|S \cap I| \mid (I_0, V_0), (I_1, V_1), \ldots, 
(I_{i_j}, V_{i_j})],
\end{eqnarray*}
where the expectation is conditioned on the partial outcome of the 
experiment up to the construction of $I_{i_j}$ and $V_{i_j}$.
We have 
\begin{eqnarray*}
Z_m & = &|S \cap I|,\\
Z_0 & = & \ex[|S \cap I|] = \frac{|S|}{2(d + 1)}, 
\end{eqnarray*}
and, for $0 \leq j < s$,
\begin{eqnarray*}
Z_j = \ex[Z_{j + 1} \mid (I_0, V_0), (I_1, V_1), \ldots, 
(I_{i_j}, V_{i_j})],
\end{eqnarray*}
where the expectation is conditioned similarly to the above.
Therefore, the sequence $Z_0$, \ldots, $Z_m$ is a martingale.

We show that
\begin{eqnarray}
\label{eqn:bounded_diff}
|Z_j - Z_{j-1}| \leq 1
\end{eqnarray}
holds for $0 < j \leq m$. 
We have 
\begin{eqnarray*}
Z_j & = & |S \cap I_{i_j}| + 
\sum_{v \in S \cap V_{i_j}} \prob(v \in I \mid v \in V_{i_j} )\\
& = & |S \cap I_{i_j}| + \frac{|S \cap V_{i_j}|(s - i_j)}{n_{i_j}}.
\end{eqnarray*}
Since both $|S \cap V_{i_j}|$ and the fraction $(s - i_j) / n_{i_j}$ 
are monotone non-increasing in $j$ and $|S \cap I_{i_j}| - |S \cap I_{i_{j-1}}|
\leq 1$ by the second condition in the definition of $i_j$, we have $Z_j -
Z_{j-1} \leq 1$.
We also have 
\begin{eqnarray*}
\frac{s - i_{j-1}}{n_{i_{j - 1}}} -
\frac{s - i_j}{n_{i_j}} & \leq &  
\frac{i_j - i_{j-1}}{n_s} \\
& \leq &  \left(\frac{s}{2|S|} + 1\right)\frac{1}{(d + 1)s}\\
& \leq & \frac{1}{2(d + 1)|S|} + \frac{1}{(d + 1)s} \\
& \leq & \frac{1}{2(d + 1)|S|} + \frac{1}{2(d + 1)|S|} \\
& \leq & \frac{1}{2|S|}\\
\end{eqnarray*}
by the third condition in the definition of $i_j$ and
\begin{eqnarray*}
\label{eqn:diffx}
|S \cap V_{i_{j-1}}| - |S \cap V_{i_j}| \leq d + 1
\end{eqnarray*}
by the second condition. Therefore, we have
\begin{eqnarray*}
Z_{j-1} - Z_j & \leq &\frac{|S \cap V_{i_{j-1}}|(s - i_{j-1})}{n_{i_{j-1}}}
- \frac{|S \cap V_{i_j}|(s - i_j)}{n_{i_j}} \\
& = & \frac{(|S \cap V_{i_{j-1}}|- 
|S \cap V_{i_j}|)(s - i_{j-1})}{n_{i_{j - 1}}}\\
&& + 
|S \cap V_{i_j}|\left(\frac{s - i_{j-1}}{n_{i_{j -1}}} - \frac{s -
i_j}{n_{i_j}}\right)\\
& \leq & \frac{(d + 1)(s - i_{j-1})}{n_{i_{j-1}}} + 
\frac{|S \cap V_{i_j}|}{2|S|}\\
& \leq & \frac{s - i_{j - 1}}{2s - i_{j - 1}} + 
\frac{|S \cap V_{i_j}|}{2|S|}\\
& \leq & 1
\end{eqnarray*}
and hence (\ref{eqn:bounded_diff}).

Applying Azuma's inequality for this martingale with 
$\lambda = t / \sqrt{m}$, we have 
\begin{eqnarray*}
\prob(Z_m > Z_0 + t) & < &  \exp\left(-\frac{t^2}{2m}\right)\\
& \leq & \exp\left(-\frac{t^2}{6|S|}\right)
\end{eqnarray*}
and
\begin{eqnarray*}
\prob(Z_m < Z_0 - t) & < &  \exp\left(-\frac{t^2}{6|S|}\right),
\end{eqnarray*}
finishing the proof of Theorem~\ref{thm:sample-indep}.

\end{document}